\documentclass[letterpaper]{article} 
\usepackage[]{aaai25}  
\usepackage{times}  
\usepackage{helvet}  
\usepackage{courier}  
\usepackage[hyphens]{url}  
\usepackage{graphicx} 
\usepackage{natbib}  
\usepackage{caption} 
\usepackage{subcaption}
\usepackage{booktabs}       
\usepackage{amsthm}
\usepackage{amsmath,amssymb,amsfonts} 
\usepackage{cleveref}
\usepackage{nicefrac}       
\usepackage{microtype}      
\usepackage{xcolor}         
\usepackage{algorithm} 
\usepackage{algpseudocode}
\usepackage{graphicx}
\usepackage{physics}
\usepackage{subcaption}
\usepackage{thmtools}
\usepackage{svg}
\usepackage{thm-restate}

\usepackage{enumitem}
\usepackage{fixbib}

\newtheorem{thm}{Theorem }
\newtheorem{pro}{Proposition}
\newtheorem{cor}{Corollary}
\newtheorem{lem}{Lemma }
\newtheorem{dnt}{Definition}
\newtheorem{rem}{Remark }
\newtheorem{assumption}{Assumption}

\newtheorem*{theorem*}{Theorem}

\usepackage{newfloat}
\usepackage{listings}
\DeclareCaptionStyle{ruled}{labelfont=normalfont,labelsep=colon,strut=off} 
\floatstyle{ruled}
\newfloat{listing}{tb}{lst}{}
\floatname{listing}{Listing}
\title{Stochastic Reinforcement Learning with Stability Guarantees\\ for Control of Unknown Nonlinear Systems}
\author {
    Thanin Quartz\textsuperscript{\rm 1},
    Ruikun Zhou\textsuperscript{\rm 1},
    Hans De Sterck\textsuperscript{\rm 1},
    Jun Liu \textsuperscript{\rm 1} 
}
\affiliations {
    \textsuperscript{\rm 1}University of Waterloo\\
    \{tquartz, ruikun.zhou, hans.desterck, j.liu\}@uwaterloo.ca
}

\usepackage{bibentry}

\begin{document}
\maketitle
\begin{abstract}
Designing a stabilizing controller for nonlinear systems is a challenging task, especially for high-dimensional problems with unknown dynamics. Traditional reinforcement learning algorithms applied to stabilization tasks tend to drive the system close to the equilibrium point. However, these approaches often fall short of achieving true stabilization and result in persistent oscillations around the equilibrium point. In this work, we propose a reinforcement learning algorithm that stabilizes the system by learning a local linear representation of the dynamics. The main component of the algorithm is integrating the learned gain matrix directly into the neural policy. We demonstrate the effectiveness of our algorithm on several challenging high-dimensional dynamical systems. In these simulations, our algorithm outperforms popular reinforcement learning algorithms, such as soft actor-critic (SAC) and proximal policy optimization (PPO), and successfully stabilizes the system. To support the numerical results, we provide a theoretical analysis of the feasibility of the learned algorithm for both deterministic and stochastic reinforcement learning settings, along with a convergence analysis of the proposed learning algorithm. Furthermore, we verify that the learned control policies indeed provide asymptotic stability for the nonlinear systems.
\end{abstract}

\section{Introduction}
In recent years reinforcement learning (RL) has achieved considerable success in complex games \cite{Silver2016-ab}, \cite{Vinyals2019-lm} and has received attention from the control community due to its ability to learn controllers in complex high-dimensional control tasks \cite{Lillicrap2015-pt}, \cite{Schulman2017-ap}. A central problem in this context is stability. In RL, stability is closely related to safety, and safety requirements that cannot be directly learned during the training phase reduce the reliability of these algorithms when applied to real-world safety-critical problems outside of simulations or lab environments. To address these challenges in safety-critical tasks, researchers have proposed solutions such as filters designed to block unwanted actions, which may involve human-based design \cite{Saunders2017-vu} or formal guarantees \cite{Alshiekh2018-dk}. A detailed overview of techniques in safe RL can be found in \cite{Pecka2014-la} and \cite{Garcia2015-el}. \newline

\noindent In the control community, stability guarantees can be derived in the linear case and its properties are well studied. Unfortunately, designing a stabilizing controller for nonlinear systems and estimating its region of attraction is one of the most challenging tasks, especially for systems with uncertainties \cite{matallana_estimation_2010}. In this setting, it is common to design a controller for the linearized system around some neighborhood of the equilibrium point. While this formulation provides local guarantees, it restricts stability and performance to a small neighborhood around the equilibrium point. In cases where the system is unknown or partially known, finding a stabilizing controller using classical methods becomes nearly impossible. As a result, techniques that rely on initial stabilizing controllers, such as adaptive dynamic programming~\cite{jiang2017robust}, may not be feasible for optimal control design. To overcome the limitations associated with local guarantees, we propose a reinforcement learning algorithm to integrate local stabilization directly into the neural policy. Our algorithm assumes no prior knowledge of the system dynamics. To achieve local stabilization, we iteratively estimate the local state-space representation and compute the local optimal gain matrix using LQR. In contrast to many classical methods, we do not assume knowledge of an initial stabilizing policy or a Lyapunov function as the process of learning and verifying a Lyapunov function is often complex and challenging \cite{Zhou2022-yt}.

\subsection{Related Work}
The idea of combining reinforcement learning with control theory has received considerable attention in recent years. For example, the work by \cite{Zanon2021-bt} and \cite{Beckenbach2019-dq} combine reinforcement learning with techniques from stabilizing model predictive control (MPC). Nevertheless, the primary focus of merging techniques in control theory and reinforcement learning has been on the development of Lyapunov-based reinforcement learning methods. A detailed account on the aforementioned methods can be found in \cite{Osinenko2022-ax}. However, despite recent progress, stabilization remains a challenging task, and these approaches often assume the existence of a local stabilizing controller and typically test their applicability on low-dimensional systems.  \newline 

\noindent Initial work combining Lyapunov theory and RL was proposed in \cite{Perkins2003-jh} where a control policy learns to switch among a number of base-level controllers and are designed using Lyapunov domain knowledge. More recently, Lyapunov theory has been employed for safe reinforcement learning, as demonstrated in \cite{Berkenkamp2016-cc}, with subsequent model-based advancements detailed in \cite{Berkenkamp2017-vo}. An approach to local stabilization, assuming knowledge of the model dynamics, is discussed in \cite{Zoboli2021-ot}. In this approach, a global controller is learned using an actor-critic method, and the LQR optimal gain controller is designed to linearly dominate around the equilibrium point. The paper also formulates conditions for ensuring local stability in their policy design. Another approach is presented in \cite{Han2020-mc} where the authors use a Lyapunov function as the critic for guaranteeing stability in the mean cost of the learnt policy and prove stabilization under high probability. An alternative constrained optimization approach for the actor's objective function in the actor-critic method was proposed in \cite{OSINENKO20208043}. The authors demonstrated practical semi-global stabilization of the sample hold closed-loop controller. The same authors later applied a similar technique to constrain the actor policy for systems with unknown parameters, establishing closed-loop stability guarantees as discussed in \cite{GOHRT20208157}.  

\subsection{Our Contribution}
We summarize the main contributions of the paper as follows:
\begin{itemize}
    \item We propose an algorithm that learns to stabilize unknown high dimensional nonlinear systems without knowledge of an initial stabilizing controller. The main component of the algorithm integrates the learned gained matrix directly into the neural policy. The algorithm is then shown to stabilize several challenging nonlinear systems. 
    \item We prove that reinforcement learning algorithms typically only guarantee practical stability and we prove the convergence for our learning algorithm to an asymptotically stabilizing neural policy. 
    \item We verify that the learned neural policy indeed stabilizes the systems using an SMT solver-based toolbox in the sense of Lyapunov stability.
\end{itemize}
 
\section{Preliminaries}
\subsection{Notation}
The following notations will be used in the paper. \newline 

\noindent Let $\mathbb{R}:=(-\infty, \infty), \mathbb{R}_{\geq 0}:=[0, \infty)$, $\mathbb{N}:=\{1,2, \ldots\}$. A function $\alpha: \mathbb{R}_{\geq 0} \rightarrow \mathbb{R}_{\geq 0}$ is of class $\mathcal{K}$ if it is continuous, zero at zero and strictly increasing, and it is of class $\mathcal{K}_{\infty}$ if, in addition, it is unbounded. A continuous function $\alpha: \mathbb{R}_{\geq 0}^2 \longrightarrow \mathbb{R}_{\geq 0}$ is of class $\mathcal{K} \mathcal{L}$ if for each $t \in \mathbb{R}_{\geq 0}, \alpha(\cdot, t)$ is of class $\mathcal{K}$, and, for each $s>0, \alpha(s, \cdot)$ is decreasing to zero. The Euclidean norm of a vector $x \in \mathbb{R}^n$ is denoted by $\|x\|$. On a compact set $D$ we denote the $\infty$-norm $\|f\|_D := \sup_{x \in D} |f(x)|$. Let $P$ be a real, square, and symmetric matrix, $\lambda_{\max }(P)$ and $\lambda_{\min }(P)$ are respectively the largest and the smallest eigenvalue of $P$.

\subsection{Dynamic Programming}
Throughout this work we consider discrete dynamical systems generating a sequence of states $\{x_k \}_{k=1}^\infty$, $x \in D \subset \mathbb{R}^n$, where $D$ is a compact and convex set, starting from some initial state $x_0 \in \mathbb{R}^n$ by the nonlinear control system 
\begin{equation} \label{dyn}
    x_{k+1} = f(x_k, u_k)
\end{equation}
where $\{u_k\}_{k=1}^\infty$, $u \in \mathcal{U} \subset \mathbb{R}^m$ are a sequence of control inputs. Moreover, we assume that the right hand side of (\ref{dyn}) is Lipschitz to guarantee a unique solution.

\begin{assumption}\label{Assumption_Lipschitz}
    Suppose that the vector field $f(x,u)$ is Lipschitz. That is, there exists constants $L_x$, $L_u$ such that 
    \begin{equation}
        \|f(x, u) - f(y, w) \| \leq L_x \| x - y \| + L_u \|u - w \|,
    \end{equation}
    for all $x,y \in D$ and $u, w \in \mathcal{U}$.
\end{assumption}

\noindent As is common for many tasks suitable for RL, we assume that we do not have explicit knowledge of the right hand side of (\ref{dyn}). For this discrete system, denote an infinite sequence of control inputs as $\mathbf{u} = (u_1, u_2, \hdots)$ and consequently denote the solution at the $k^\text{th}$ time state with initial condition $x_0$ as $\Psi(k, x_0, \mathbf{u}(k))$ where $\mathbf{u}(k) = u_k$. The function $\Psi$ is called the \textit{flow map}. Given a reward function $r : \mathbb{R}^n \times \mathbb{R}^m \to \mathbb{R}$ we define the discounted objective function.

\begin{dnt}[Objective Function] \label{objective}
The $\gamma$-discounted objective function, $\mathbf{J}_\gamma^{\mathbf{u}} : \mathbb{R}^n \to \mathbb{R}$ is defined as 
\begin{equation} \label{reward_flow}
    \mathbf{J}_\gamma^{\mathbf{u}}(x_0) := \sum_{k=0}^\infty \gamma^k r(\Psi(k, x_0, \mathbf{u}(k)), \mathbf{u}(k)). 
\end{equation}
where $\gamma \in (0, 1)$. 
\end{dnt}

\begin{rem}
    In this setting where transitions are deterministic, the definitions of the objective function and the value function coincide.
\end{rem}

\noindent The objective of reinforcement learning is to identify a controller $\mathbf{u}$ that maximizes the objective function.

\begin{dnt}[Optimal Value Function]\label{optimal_value}
    The optimal value at a state $x \in \mathbb{R}^n$ is denoted as $V_\gamma(x_0)$ and is defined as 
    \begin{equation}
        V_\gamma(x_0) := \sup_{\mathbf{u}} \mathbf{J}^{\mathbf{u}}_\gamma(x).
    \end{equation}
\end{dnt}

\noindent However, in most RL algorithms, estimating the \textit{action-value function} leads to greater performance in the model free setting. 

\begin{dnt}[Action-Value Function]
\label{action_value}
    The action value function $\textbf{Q}^{\mathbf{u}} : \mathbb{R}^n \times \mathbb{R}^m \to \mathbb{R}$ is defined as 
    \begin{equation}
        \textbf{Q}^{\mathbf{u}}(x, u) = r(x, u) + \mathbf{J}^\mathbf{u}_\gamma (y)
    \end{equation}
    where $y = f(x, u)$. That is, the cumulative reward at state $x$ by taking action $u$ and subsequently following the control sequence $\mathbf{u}$ at the state $y$.
\end{dnt}

\subsection{Infinite Horizon LQR and Stability}
\noindent Consider the discrete time linearization of $f(x, u)$ 

\begin{equation} \label{liner_dyn}
    x_{n+1} = A x_n + B u_n
\end{equation}

\noindent about the equilibrium point where we can assume without loss of generality that the origin is the equilibrium point. The infinite horizon LQR problem seeks a stablizing controller $\textbf{u}$ that minimizes a quadratic cost function, thereby ensuring optimal performance and stability of the closed-loop system. The cost in control is dual to the reward formulation in RL. In particular, the cost has the form 

\begin{equation} \label{quadratic_cost}
    \textbf{J}^{\text{LQR}} = \sum_{n=0}^\infty \left( x_n^T Q x_n + u_n^T R u_n \right)
\end{equation}

\noindent where $Q$ is a positive definite matrix $R$ is a positive semi definite matrix. In \cite{Bertsekas2012-ot} the optimal policy $\mathbf{u}^{\text{LQR}}$ is derived for state system (\ref{liner_dyn}) with cost (\ref{quadratic_cost}) and is given by the feedback controller $\mathbf{u}^{\text{LQR}}(x) = -K x$ where

\begin{equation}\label{LQR_K}
    K = \left(R+ B^T P B\right)^{-1} B^T P A
\end{equation}

\noindent and $P$ satisfies the discrete algebraic Ricatti equation (DARE)
\begin{equation}\label{LQR_P}
    P = Q + A^T\left(P - P B\left(R + B^T P B\right)^{-1} B^T P\right) A.
\end{equation}

\noindent It is also well known that for linear dynamics, the controller $\mathbf{u}(x) = -Kx$ is a stabilizing controller in the sense of the following definition.

\begin{dnt}[Asymptotic Stability]
An equilibrium point for the closed loop system 
\begin{equation}\label{closedloop_dyn}
    x_{k+1} = f(x_k, \mathbf{u}(x_k))
\end{equation}
is said to be asymptotically stable, if the following conditions are satisfied: \newline
(1) (Lyapunov stability) For every $\epsilon>0$, there exists a $\delta>0$ such that $\|x_0\|<\delta_{\epsilon}$ implies that $x_k$ is defined for all $k$ and $\|x_k\|<\epsilon$ for all $k$; and \newline 
(2) (Attractivity) There exists some $\rho>0$ such that, for every $\epsilon>0$, there exists some $N \in \mathbb{N}$ such that $x_k$ is defined for all $k \geq N$ and $\|x_k\|<\epsilon$ whenever $\|x_k\|<\rho$ and $k \geq N$.
\end{dnt} 

\subsection{Soft Actor Critic}
The soft actor critic algorithm \cite{Haarnoja2018-ha} is a popular reinforcement learning algorithm that has been successfully applied to continuous control tasks. In the RL setting, when a neural network parameterizes the controller, we refer to the controller as a \textit{policy}, denoted by $\pi_\phi$, to differentiate it from a general admissible controller $\mathbf{u}$. Additionally, we consider a probabilistic distribution for $\pi_\phi$ which is assumed to belong to a Gaussian family 

\begin{equation}\label{policy}
    \pi_\phi(x) \in \{\mathcal{N}(\mu(x), \Sigma(x)) \},
\end{equation}

\noindent where the mean $\mu$ and covariance matrix $\Sigma$ are functions of the state. The SAC objective adds an entropy term $\mathcal{H}\left(\pi\left(\cdot \mid \mathbf{x}\right)\right)$ into the objective function to learn a stochastic policy that promotes exploration 

\begin{equation}\label{obj_entropy}
\begin{split}
& J_\gamma(x_0, \pi)= \\
& \sum_{k=0}^\infty \gamma^k \mathbb{E}_{\left(\mathbf{x}_k, \mathbf{a}_k\right) \sim \rho_\pi}\left[r\left(\mathbf{x}_k, \mathbf{a}_k\right)+\alpha \mathcal{H}\left(\pi\left(\cdot \mid \mathbf{x}_k\right)\right)\right],
\end{split}
\end{equation}

\noindent where $\alpha$ is called the \textit{temperature} parameter which controls the stochasticity of the policy and $\rho_\pi$ is the trajectory distribution of the policy $\pi_\phi$. Furthermore, the SAC algorithm estimates the value function $V^{\pi_\phi}_\psi(x)$, action-value function $Q^{\pi_\phi}_\theta(x, a)$ and learns an actor function $\pi_\phi(x)$ which are all approximated with neural networks. The value network is trained to minimize the squared residual error 

\begin{equation}\label{Value_loss}
J_V(\psi)=\mathbb{E}_{x_k \sim \mathcal{D}}\left[\frac{1}{2}\left(V_\psi\left(x_k\right) - \hat{V}\right)^2\right],
\end{equation}

\noindent where $\mathcal{D}$ is the distribution of previously sampled states and actions, or a replay buffer with the target value network

\begin{equation}
    \hat{V}(x_k) = \mathbb{E}_{a_k \sim \pi_\phi}\left[Q_\theta\left(x_k, a_k\right)-\log \pi_\phi\left(a_k \mid x_k\right)\right].
\end{equation}

\noindent The soft Q function is trained to minimize the soft Bellman residual as its loss function 

\begin{equation}
\begin{split}
&J_Q(\theta)= \\
&\mathbb{E}_{\left(x_k, a_k\right) \sim \mathcal{D}}\left[\frac{1}{2}\left(Q_\theta\left(x_k, a_k\right)-\hat{Q}\left(x_k, a_k\right)\right)^2\right],
\end{split}
\end{equation}

\noindent with the target action-value network

\begin{equation} 
\hat{Q}\left(x_k, a_k\right)=r\left(x_k, a_k\right)+\gamma \mathbb{E}_{\mathbf{x}_{k+1} \sim p}\left[V_{\bar{\psi}}\left(x_{k+1}\right)\right].
\end{equation}

\noindent where $p$ is the distribution of the next state $x_{k+1}$. Finally, we train the policy network to directly minimize the following expected KL divergence

\begin{equation} \label{actor_KL}
\begin{split}
&J_\pi(\phi)= \\
&\mathbb{E}_{x_k \sim \mathcal{D}}\left[\mathrm{D}_{\mathrm{KL}}\left(\pi_\phi\left(\cdot \mid x_k\right) \| \frac{\exp \left(Q_\theta\left(x_k, \cdot\right)\right)}{Z_\theta\left(\mathbf{x}_k\right)}\right)\right] .
\end{split}
\end{equation}

\noindent The intuition behind this formula is that higher values of $Q_\theta(x_n, u_n)$ signal higher cumulative rewards so the policy will be proportional to the current estimate of the soft action-value function. A more computable form for (\ref{actor_KL}) is the following

\begin{equation} \label{actor_loss}
\mathbb{E}_{x \sim \mathcal{D}}\left[\mathbb{E}_{a \sim \pi_{\phi(\cdot \mid x)}}\left[\log \pi_\phi(a \mid x)-Q_\theta(x, a)\right]\right].
\end{equation}

\section{Reinforcement Learning and Stability}
In this section we describe the RL algorithm for learning the policy $\pi$ as well simultaneously learning local linear dynamics for achieving stability. In particular, we focus on the Soft Actor Critic algorithm for our purposes, but the formulation with any reinforcement learning algorithm is equivalent. Since we focus on stability, the reward function we implement is set to 
\begin{equation} \label{reward}
    r(x_n, u_n) = R - x^T Q x
\end{equation}
where $R$ is a positive constant and $Q$ is a positive definite matrix. The design of the reward function will encourage the agent to get close to the equilibrium point to maximize the cumulative rewards. We could also penalize the actions as this will not affect any of the theoretical results derived in the following section, however since we already consider constraints we omit this penalization to simplify the training process. Additionally, we suppose that the reward function satisfies the \textit{terminal condition} on the domain $D$.
\begin{dnt}[Terminal Condition]
    The reward function is said to satisfy the terminal condition on a domain $D$ for some $c \geq 0$ if 
    \begin{equation} \label{terminal}
    r(x_N, u_N) = -c, 
    \end{equation}
    where $N$ is the first iteration such that $x_N \notin D$ and
    \begin{equation*}
        r(x_n, u_n) = 0 \; \text{ for all } n > N.
    \end{equation*}
\end{dnt}

\noindent We design the reward/environment to satisfy the terminal condition to ensure that the agent collects data near the equilibrium point. We now describe the components of the algorithm. For the full details please refer to Algorithm 1 in the Appendix. \newline 

\noindent \textbf{Step 1}. As a first step prior to solving the DARE we need to model the local linearization at the equilibrium point. While the reinforcement learning algorithm is running we collect data points whenever the agent is within a fixed $\eta > 0$ distance from the origin. We learn a state space representation 
\begin{equation} \label{learned_linear_dynamics}
    x_{n+1} = \hat{A} x_n + \hat{B} u_n
\end{equation}
by collecting data points for $x_n$ and $x_{n+1}$ whenever $x_n \in B_{\eta}(0)$ in two different ways. Until a preset threshold has been reached, take action $u_n = 0$ and insert $x_n$ and $x_{n+1}$ as row vectors in the data matrices $X$ and $Y$ respectively. Then learn $\hat{A}$ as the least squares estimate 
\begin{equation} \label{learnA}
\hat{A} = \arg \min _A\|Y-A X\|_2.
\end{equation}
Then once this is completed, we similarly collect data matrices $X, Y$ and $U$ where $U$ stores the non-zero actions and learn B from 
\begin{equation} \label{learnB}
\hat{B} = \arg \min _B\|Y-\hat{A} X - B U\|_2.
\end{equation}

\noindent \textbf{Step 2}. After learning a representation for ($\ref{learned_linear_dynamics}$) we calculate the gain matrix $K$ by solving the DARE. \newline

\noindent \textbf{Step 3}. Then we update the empirical loss function of (\ref{actor_loss}) as
\begin{equation} \label{gain_actor_loss}
    \begin{split}
        \sum_{x_m} \sum_{a_k} \left( \log \pi_\phi(a_k \mid x_m)-Q_\theta(x_m, a_k) \right) \\
        + \| \nabla \mu(\pi_\phi) - K \|_2 + \|\mu(\pi_\phi)(0) \|
    \end{split}
\end{equation}
where $\mu(\pi_\phi) : \mathbb{R}^n \to \mathbb{R}$ maps the state to the mean of the policy and $\nabla \mu(\pi_\phi)$ is the Jacobian of this function with respect to the state $x$. By minimizing this loss, the policy will learn to be locally stabilizing with its linear approximation close to the optimal gain matrix $K$. 
\newline

\noindent \textbf{Step 4}. After training has completed, set 
\begin{equation}
    \pi_\phi(x) \leftarrow \pi_\phi(x) - \pi_\phi(0),
\end{equation}
as this will ensure that $\pi_\phi(0) = 0$ and that the origin is a stable equilibrium point. \newline 

\noindent \textbf{Step 5}. When evaluating the policy for stabilization tasks, the action taken should be the mean of the policy, that is we should set 
\begin{equation}
    \pi_\phi(x) \leftarrow \mu(\pi_\phi(x))
\end{equation}
as a stochastic policy will only stabilize with high probability should the variance decrease to zero about the equilibrium point. \newline 

\begin{figure*}[ht!]
  \centering
  \includegraphics[scale=0.50]{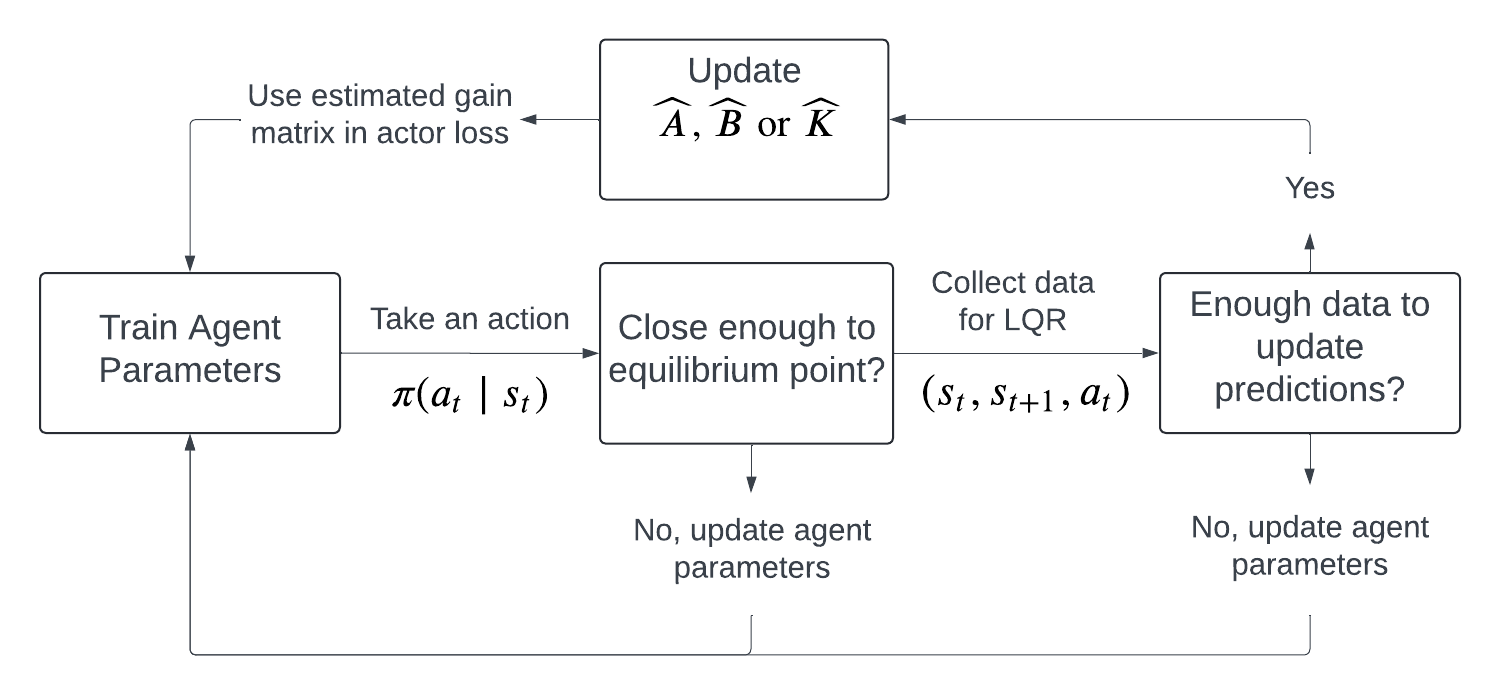}
  \caption{Overall learning algorithm of the proposed method.}
  \label{learning_algorithm}
\end{figure*}

\noindent We found that the training was more efficient when collecting data while simultaneously training all the parameterized functions according to (\ref{Value_loss}) - (\ref{actor_loss}) rather than first learning the local dynamics and updating the actor parameters according to (\ref{gain_actor_loss}). This is because if knowledge of local dynamics is used, most states will be close to the equilibrium point and the performance will not generalize well outside a neighbourhood of the equilibrium point.

\subsection{Enforcing Action Bounds}
Since the Gaussian distribution has infinite support, actions need to be mapped to a finite domain. To address this, we use a squashing function that allows for more general actions. In particular, supposing that $\mathbf{w}$ is a random variable with support $\mathbb{R}^d$, the random variable 
\begin{equation}\label{squashing}
    \mathbf{y} = A \tanh(\mathbf{w}) + b    
\end{equation}
where $A$ is an invertible matrix, $b$ is a vector and the hyperbolic tangent, $\tanh$, is applied element-wise has a support over a parallelepiped in $\mathbb{R}^d$. We will apply  the change of variables formula to compute the log likelihoods of the bounded actions. Let $\mathbf{w}$ be a random variable with Gaussian distribution $\mu(\mathbf{w}|x)$ and $\mathbf{y} = A \tanh(\mathbf{w}) + b$. Then writing $\mathbf{z} = \tanh(\mathbf{w})$ we have that $J_y = A$ and $J_z = \text{diag}(1 - \tanh^2(w_1), \hdots, 1 - \tanh^2(w_d))$, where $J_y$ is the Jacobian of $\mathbf{y}$ with respect to $\mathbf{z}$ and $J_z$ is the Jacobian of $z$ with respect to $x$. From the change of variables formula,
\begin{align*}
    \pi(y|x) &= \mu(\mathbf{w}|x) \cdot |\det(A)|^{-1} |\det J_z|^{-1}
\end{align*}
and in the case when $A$ is a diagonal matrix, the log likelihood has a convenient form
\begin{equation}
    \pi(y|x) = \frac{\mu(\mathbf{w}|x)}{\prod_{i = 1}^n a_{ii}(1 - \tanh^2(w_i))},
\end{equation}
and the log-likelihood is given by 
\begin{equation}
    \log \pi(y|x) = \log \mu(\mathbf{w}|x) - \sum_{i=1}^n \log  (a_{ii}(1 - \tanh^2(w_i))).
\end{equation}

\section{Theoretical Guarantees}
Recall the system (\ref{dyn}):
\begin{align*}
    x_{n+1} = f(x_n, u_n),
\end{align*}with state $x \in \mathbb{R}^n$ and input $\mathcal{U}(x) \subset \mathbb{R}^m$ where $\mathcal{U}(x)$ is a set of admissible controls (as in \cite{Keerthi1985-py}) associated to a state $x$. Since stabilization under discounting in the discrete setting is much more challenging than without a discount factor, we make the following assumptions.

\begin{assumption}\label{Assumption_optimal}
    (a) For any initial condition $x_0 \in \mathbb{R}^n$ there exists an infinite length control system $\mathbf{u}^\star$ such that 
    \begin{equation} \label{value_function}
         V_\gamma(x) = \mathbf{J}^{\mathbf{u}^\star}_\gamma(x_0) = \sup_{\mathbf{u}} \mathbf{J}^{\mathbf{u}}_\gamma(x).
    \end{equation}
    (b) Suppose that there exists a constant $a_V$ such that for any $\gamma \in (0, 1)$ and $x \in \mathbb{R}^n$
    \begin{equation}\label{Assumption_linearvalue}
        \frac{R}{1 - \gamma} - V_\gamma(x) \leq a_V \|x \|^2.
    \end{equation}
\end{assumption}

\noindent For (a), this assumption states that for any initial condition $x_0 \in \mathbb{R}^n$ the optimal value function in Definition \ref{optimal_value} exists. Conditions on the vector field in (\ref{dyn}) and the rewards to ensure that Assumption~\ref{Assumption_optimal} holds are given in \cite{Keerthi1985-py}. For (b) a sufficient condition is for the optimal policy to be exponentially stabilizable. That is, there exists $M > 0$ and $\lambda > 0$ such that for each $x \in \mathbb{R}^n$ there exists an infinite length control sequence $\mathbf{u}(x)$ that satisfies $R - r(\Psi(k, x, \mathbf{u}(k)), \mathbf{u}(k)) \leq M \|x \|^2 e^{-\lambda_d}k$. For weaker conditions that guarantee condition (b) please refer to \cite{Grimm2005-cm}.

\begin{assumption}\label{Assumption_policy}
    Suppose that the policy in (\ref{policy}) is parameterized by a neural network with smooth activation functions. 
\end{assumption}

\noindent This assumption is necessary since in the proofs we need to approximate the value function with a neural policy and it always holds in our implementation as we use the hyperbolic tangent activation. Furthermore, as we shall see in the next section, reinforcement learning methods typically guarantee practical stability in the discrete-time with discounting setting and therefore, the following assumption will allow for asymptotic stability.

\begin{assumption} \label{Assumption_LQR}
    Let $K$ be the optimal gain matrix of the solution to the local LQR problem (\ref{LQR_K}). Suppose that for any $\epsilon_K > 0$, our learning algorithm converges to an estimate of the optimal gain matrix $\hat{K}$ such that $\| K - \hat{K} \|_2 < \epsilon_K$.
\end{assumption}

\noindent This assumption is reasonable because, in practice, learning the linear state-space representation from data is often straightforward, sometimes even achievable with a single trajectory \cite{hu2022sample}. Additionally, we can also empirically test whether the learned gain matrix provides stabilization within a neighborhood around the origin.

\subsection{Deterministic Guarantees}
We first formulate theoretical guarantees in the deterministic setting and then prove theoretical guarantees for our learning algorithm. This is because we consider stochasticity in the our learning algorithm algorithm for exploration, but test the effectiveness of the policy by taking the expectation. The following proposition establishes the value function is pointwise bounded and that it strictly increases along its optimal trajectories.

\begin{pro}\label{Proposition_bounds}
There exists comparison functions $\alpha_1, \alpha_2 \in \mathcal{K}_\infty$ and $\gamma^\star$ such that for any $\gamma$ satisfying $\gamma \in (\gamma^\star, 1)$
the following holds
\begin{equation}\label{bound1}
    \alpha_1(x) \leq \frac{R}{1 - \gamma} - V_\gamma(x) \leq \alpha_2(x)
\end{equation}
for any  $x \in \mathbb{R}^n$, and 
\begin{equation}\label{bound2}
    V_\gamma(x) - V_\gamma(y) \leq - c \| x \|^2
\end{equation}
for any $x \in \mathbb{R}^n$, and $y = f(x, \mathbf{u}_1(x))$.
\end{pro}

\begin{rem}
    In the above proposition we have that $\alpha_1(s) = \lambda_{\min}(Q) s^2$ and $\alpha_2(s) = a_V s^2$.   
\end{rem}

\noindent The proposition asserts that for a sufficiently large discount factor, the value function can from part of a Lyapunov function. The proof of the proposition can be found in the Appendix. Consequently, this proposition can be used to show that the optimal policy is stabilizing. 

\begin{thm}
    Suppose that $\pi^\star$ is the optimal policy for the objective (\ref{objective}). Then there exists $\gamma^\star$ such that for any $\gamma \in (\gamma^\star, 1)$ the system 
    \begin{equation}
        x_{n+1} = f(X_n, \pi^\star(x_n))
    \end{equation}
    is asymptotically stable.
\end{thm}

\noindent The proof of the theorem can be found in the Appendix. On the other hand, given that the class of policies is parameterized by neural networks, we have to restrict ourselves to these class of functions. However, due to approximation errors, only \textit{practical stability} in the sense of the following theorem is guaranteed.

\begin{thm} \label{thm_nnet_stab}
    Fix $\Delta > 0$ and let $\pi^\star$ is the optimal policy for the objective (\ref{objective}). Then, there exists $\epsilon > 0$, $\gamma^\star$ and $N \in \mathbb{N}$ such that for any $\gamma \in (\gamma^\star, 1)$ and any neural policy $\|\pi_\phi - \pi^\star \|_\infty < \epsilon$ such that the following holds 
    \begin{equation}
        \| \Psi(k, x, \pi_{\phi, k}(x)) \| \leq  \Delta \quad \forall k \geq N.
    \end{equation}
    for any $x \in K$.
\end{thm}

\begin{rem}
    An alternative formulation for practical stability is presented in the Appendix. The proof of this result relies heavily on a smooth version of the universal approximation theorem.
\end{rem}

\noindent The notion of stability is called \textit{practical} since we are only guaranteed stability in a small neighborhood of the origin. To ensure asymptotic stability our algorithm will learn a local gain matrix that guarantees local stability. This result will be formulated in the following section.


\subsection{Stochastic Guarantees} 
In this section, we prove stability and convergence guarantees for our learning algorithm. As a first result we show that the temperature parameter in (\ref{obj_entropy}) can be controlled such that the optimal solution to the SAC objective is close in expectation to the solution of the optimal deterministic policy (\ref{objective}). 
\begin{pro} \label{Proposition_temp}
    Fix $\epsilon > 0$ and consider the system (\ref{dyn}) and the SAC objective. Denote the optimal solution to (\ref{obj_entropy}) as $\pi^{\text{SAC}}$ and the the optimal solution to (\ref{objective}) as $\pi^\star$. Then there exists $\alpha^\star$ such that for any temperature $0 < \alpha < \alpha^\star$, we have
    \begin{equation}
        |\mathbb{E}\pi^{\text{SAC}} - \pi^\star |_D < \epsilon.
    \end{equation}
\end{pro}
\noindent The value of $\alpha^\star$ is included in the proof as presented in the Appendix. Since in actual implementation, $\mathbb{E}\pi^{\text{SAC}}$ is parameterized by a neural network, by Theorem \ref{thm_nnet_stab} we can guarantee practical stability for a near optimal neural policy.   

\begin{cor}\label{cor_practical}
    Suppose that $\pi^{\text{SAC}}$ is the optimal policy to (\ref{obj_entropy}). Then for any $\Delta > 0$, there exists $\epsilon, \gamma^\star, \alpha^\star > 0$, and $N \in \mathbb{N}$ such that for any $\alpha \in (0, \alpha^\star)$, $\gamma \in (\gamma^\star, 1)$ and any neural policy $\| \mathbb{E} \left[\pi^{\text{SAC}} - \pi_{\phi}{^{\text{SAC}}} \right] \|_D < \epsilon$, the following holds 
    \begin{equation}
        \| \Psi(k, x, \pi_{\phi, k}^{\text{SAC}}(x)) \| \leq  \Delta \quad \forall k \geq N.
    \end{equation}
    for any $x \in D$.
\end{cor}

\noindent However, if the learned neural policy satisfies a linear order of approximation with sufficiently small decay factor, it is possible to extend practical stability to asymptotic stability.

\begin{pro}
    Suppose that $\pi^\star$ is the optimal policy to (\ref{objective}) and $\pi^{\text{SAC}}$ is the optimal policy to (\ref{obj_entropy}). Then, there exists $\epsilon, \gamma^\star, \alpha^\star > 0$ such that if the neural policy that satisfies the approximation error $\| \mathbb{E} \left[\pi^{\text{SAC}} - \pi_{\phi}{^{\text{SAC}}} \right] \|_D < \epsilon$ also satisfies the order of approximation
    \begin{equation} \label{linear_order}
        | \mathbb{E} \left[\pi^\star(x) - \pi_{\phi}{^{\text{SAC}}}(x) \right] | \leq \eta \|x\|
    \end{equation}
    for any decay rate $\eta$ sufficiently small, then for any $\alpha \in (0, \alpha^\star)$ and $\gamma \in (\gamma^\star, 1)$, the neural policy $\mathbb{E} \left[\pi_{\phi}{^{\text{SAC}}} \right]$ is asymptotically stabilizing.
\end{pro}

\noindent These two previous results form the bridge between stability and reinforcement learning. The ultimate goal of stabilizing can be realized by learning a policy that approximates the optimal policy well. However, due to randomness and optimization errors that are typically intractable, in practice, reinforcement learning will only manage to learn a practically stabilizing policy as in Corollary \ref{cor_practical}. This highlights the importance of our learning algorithm. \newline 

\noindent Next we analyze the convergence of our learning algorithm. However, since the convergence of maximal entropy reinforcement learning algorithms is only guaranteed in the tabular case \cite{Haarnoja2018-ha} and the continuous case is still a challenging open problem, we make the assumption that the SAC algorithm can learn a neural policy that is arbitrarily close to the optimal one.

\begin{assumption}\label{Assumption_actorloss}
    Suppose that $\pi^{\text{SAC}}$ is the optimal policy to the SAC objective ($\ref{obj_entropy}$), then for any $\epsilon > 0$ there exists a neural policy $\pi_\phi^{\text{SAC}}$ such that $\| \mathbb{E} \left[\pi^{\text{SAC}} - \pi_{\phi}^{\text{SAC}} \right] \|_D < \epsilon$.
\end{assumption}

\noindent Indeed, this seems to hold in practice as evidenced in the following section. The neural policy learned by SAC will reach and stay in a small neighborhood around the equilibrium point. Furthermore, the robustness of the LQR optimal gain matrix asserts a sufficient close gain matrix will be stabilizing. Therefore, combining the practical stability of the near optimal RL policy and the local stabilization properties of learned gain matrix will show that our learning algorithm will converge to an asymptotically stabilizing neural policy.
\begin{thm}\label{thm_converge}
    Suppose that Assumptions \ref{Assumption_LQR} and \ref{Assumption_actorloss} hold for our learning algorithm. Then the learning algorithm converges to a asymptotically stabilizing neural policy $\mathbb{E} \left[\pi_{\phi}{^{\text{SAC}}} \right]$.
\end{thm}

\begin{figure*}[ht!]
    \centering
    \begin{subfigure}[b]{0.45\textwidth}
        \centering
        \includegraphics[width=\textwidth]{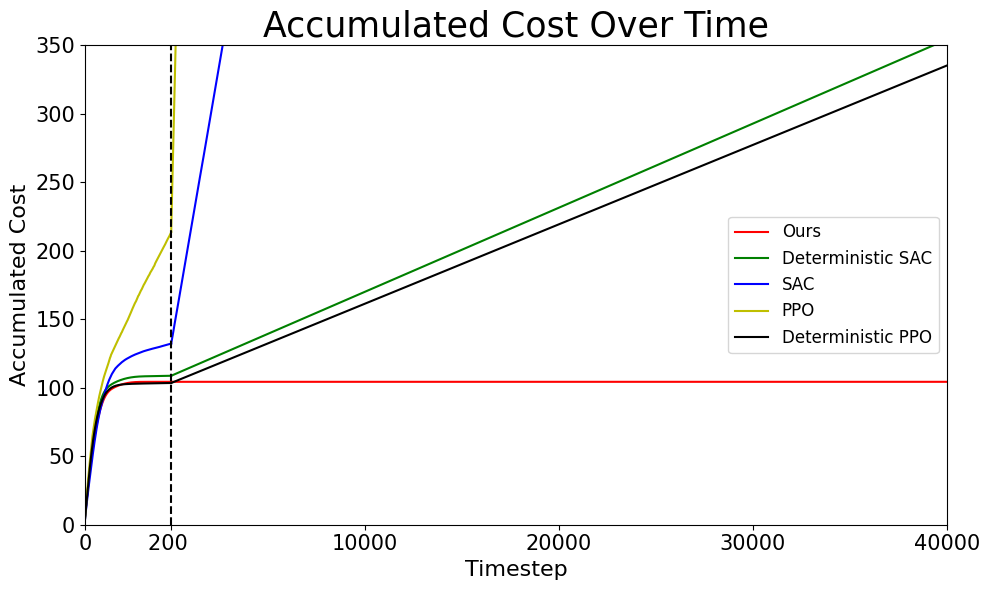}
        \caption{Inverted Pendulum}
        \label{fig:1a}
    \end{subfigure}
    \hfill
    \begin{subfigure}[b]{0.45\textwidth}
        \centering
        \includegraphics[width=\textwidth]{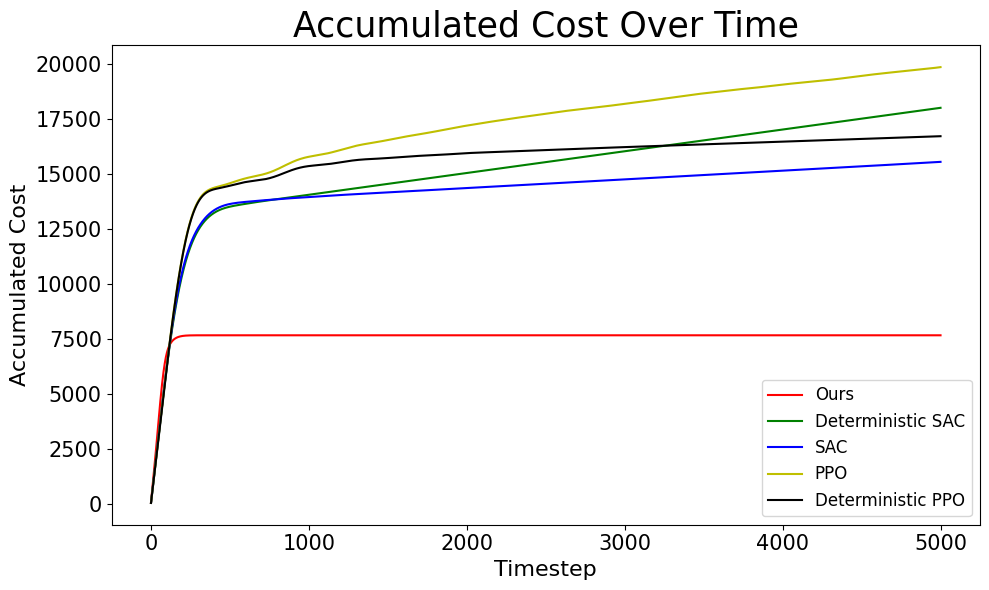}
        \caption{Cartpole}
        \label{fig:1b}
    \end{subfigure}
    
    \vspace{0.31cm}

    \begin{subfigure}[b]{0.45\textwidth}
        \centering
        \includegraphics[width=\textwidth]{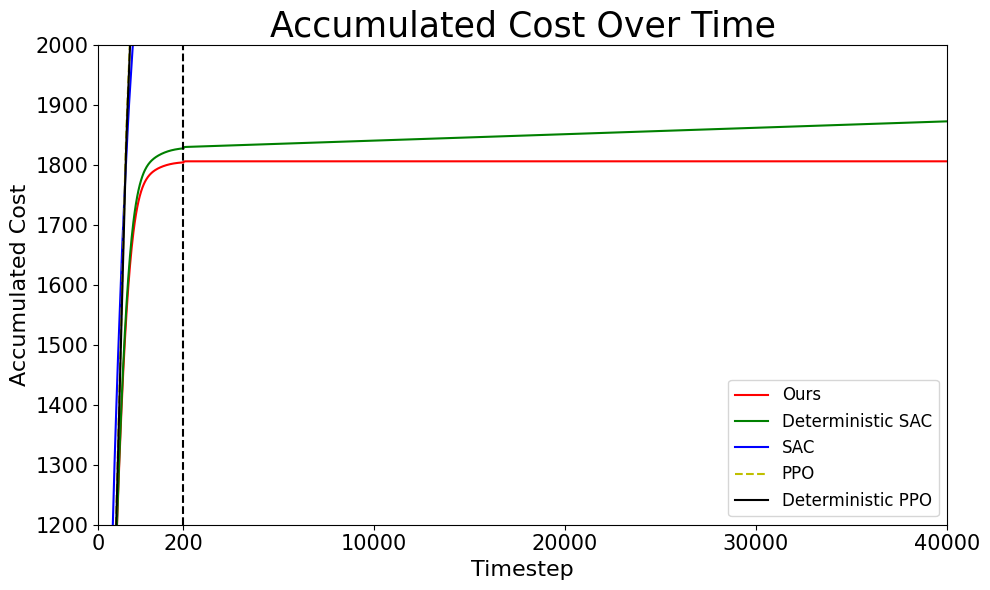}
        \caption{2D Quadrotor}
        \label{fig:1c}
    \end{subfigure}
    \hfill
    \begin{subfigure}[b]{0.45\textwidth}
        \centering
        \includegraphics[width=\textwidth]{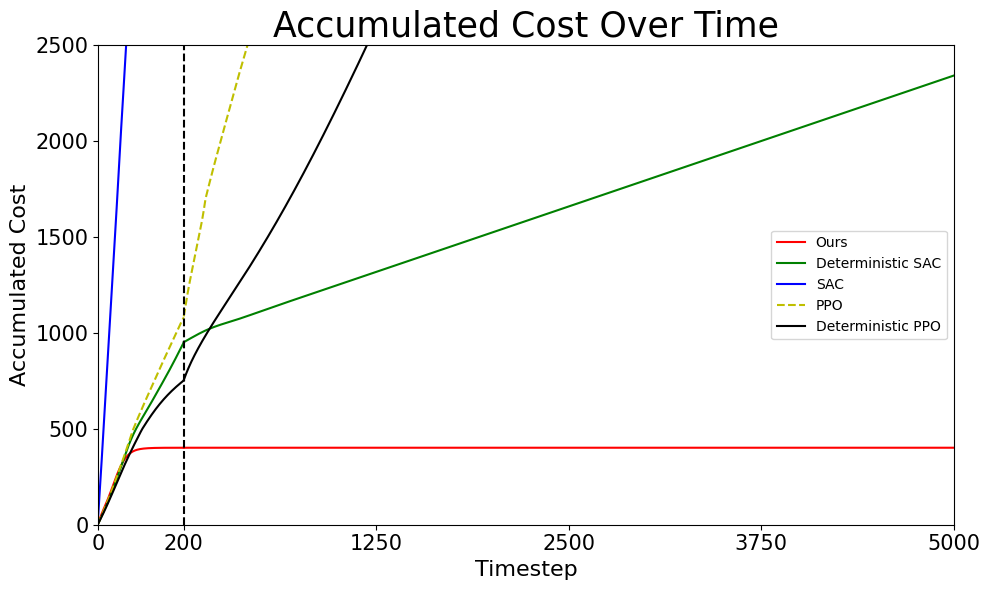}
        \caption{3D Quadrotor}
        \label{fig:1d}
    \end{subfigure}
    
    \caption{Plots of Accumulated Costs for the Four Environments}
    \label{fig_costs}
\end{figure*}

\section{Numerical Experiments}
In this section, we present numerical examples to evaluate the performance of the proposed learning algorithm. In these simulations we aim to accomplish the following goals: 1) Evaluate the performance of our learning algorithm against stochastic and deterministic versions of SAC and PPO, 2) Evaluate whether our learning algorithm can achieve good performance on high-dimensional control problems while still stabilizing and satisfying the constraints, 3) Demonstrate formal verification of the asymptotic stability for low dimensional systems with the learned neural policy. 

\subsection{Comparison with Reinforcement learning methods}
We compare our learning algorithm against deterministic and stochastic SAC and PPO. The ultimate goal of this task is to learn an asymptotically stabilizing policy. Therefore, we compare the accumulated cost over time where the cost is measured by the quadratic form $x^TQx + u^TRu$ where $Q$ and $R$ are the same matrices used to learn the optimal gain matrix for the LQR controller. This is different from the typical comparisons in RL which average the performance of the learning algorithm over the episodic training. For every environment, we evaluate the performance by uniformly sampling the initial states and test the aforementioned policies. All hyperparameters, network architectures, and costs are listed in the Appendix and the code for all comparisons are provided in the supplementary material. \newline

\noindent In general, stochastic policies perform poorly, whereas PPO performs well on low-dimensional tasks. Deterministic SAC shows strong performance across all tasks, as it can reach and maintain a position within a small neighborhood of the origin, indicating that the policy can achieve practical stability. However, for true stabilization, only our learning algorithm successfully completes this task. The plots for stabilization are shown in the Appendix. As a final statement, it is important to emphasize that our control tasks are more difficult than the typical optimal control tasks as we consider constraints, whereas standard optimal control methods such as LQR cannot handle constraints. Additionally, our algorithm is able to stabilize on larger domains as compared to the recent work in \cite{Meng2024PhysicsInformedNN}.

\subsection{Verification for Asymptotic Stability}
When the sampling time is small, the learned neural policy $\pi_\phi$ should be able to stabilize the continuous-time dynamics. In this work, we verify that the origin of the closed-loop continuous-time system $\dot{x} = f(x, \pi_\phi(x))$ is asymptotically stable within its domain for low dimensional systems such as the inverted pendulum. This verification is done using a quadratic Lyapunov function and the LyZNet toolbox~\cite{liu2024tool} built upon SMT solvers. The quadratic Lyapunov function is generated by linearizing the system and then solving the Lyapunov equation. With the quadratic Lyapunov function, we can first find a local region-of-attraction (ROA) where linearization dominates. Next, a reachability step is utilized to enlarge the ROA using the Lyapunov stability analysis for the quadratic Lyapunov function with respect to the local ROA. In doing so, all the states in the domain are guaranteed to eventually converge the origin, yielding asymptotic stability.
The verified results with the quadratic Lyapunov function and corresponding ROA estimate can be found in the Appendix.

\section{Conclusion}
In this work, we propose an algorithm for learning a stabilizing controller for unknown nonlinear systems. By integrating the local optimal gain matrix directly into the neural policy we design a policy that satisfies classical local LQR stabilization and numerically shows global stability. This approach significantly outperforms standard RL algorithms which in the theory and implementation only guarantee practical stability. Additionally, we provide theoretical analysis for stabilization and the convergence for our algorithm to an asymptotically stabilizing policy, which can be verified with LyZNet, a toolbox based on SMT solvers.  Limitations and potential future work are discussed in the Appendix.


\bibliography{mybib}
\newpage 

\onecolumn
\appendix

\begin{center}
\section{Appendix}
\end{center}
\subsection{Proofs of Theoretical Guarantees}
We recall that a $C^k(D)$ function means a $k$-times continuously differentiable function defined on a domain $D$ and a $C(D)$ function is continuous. The following lemmas will be integral to the proofs in this section.

\begin{lem}[Weierstrass M-test]
Let $\{f_n(x))\}$ be a sequence of real valued functions defined on a set $D$. Suppose there exists a sequence of non-negative constants $\{M_n\}$ such that $|f_n(x)| \leq M_n$ for all $x \in D$ and for all $n$, and that $\sum_{n=1}^{\infty} M_n$ is convergent. Then the series $\sum_{n=1}^{\infty} f_n(x)$ converges uniformly on $D$.
\end{lem}

\begin{lem}[Universal Approximation Theorem]
    Let $D \subset \mathbb{R}^m$ be a compact set and suppose $f : D \to \mathbb{R}^n \in C(\mathbb{R}^n)$. Then, for every $\epsilon > 0$ there exists a neural network $\phi : D \to \mathbb{R}$ with $C(\mathbb{R})$ non-polynomial activation function such that 
\begin{equation}
    \|f(x) - \phi(x)\|_D < \epsilon. 
\end{equation}
\end{lem}

\begin{lem}[Universal Approximation Theorem - Smooth version]
\label{UAP_der}
Let $D \subset \mathbb{R}^m$ be a compact set and suppose $f : D \to \mathbb{R}^n \in C^1(\mathbb{R}^n)$. Then, for every $\epsilon > 0$ there exists a neural network $\phi : D \to \mathbb{R}$ with $C^1(\mathbb{R})$ non-polynomial activation function such that 
\begin{equation}
\label{fcn_app}
    \|f(x) - \phi(x)\|_D < \epsilon, 
\end{equation}
and for all $i = 1, \hdots, n$, the following simultaneously holds
\begin{equation}
\label{der_app}
    \left\|\frac{\partial f}{\partial x_i} - \frac{\partial \phi}{\partial x_i} \right\|_D < \epsilon.
\end{equation}
\end{lem}

\noindent The proofs of the universal approximation theorems for neural networks can be found in \cite{pinkus_1999}. As a preliminary result we first prove the claim under Assumption \ref{Assumption_optimal} that if the controller is exponentially stabilizable then Assumption \ref{Assumption_optimal} (b) is guaranteed to hold. 

\begin{thm}[Hartman-Grobman]
    Consider the system (\ref{dyn}) with a policy $\pi(x)$. Suppose that at the equilibrium point, the Jacobian matrix $A=\left[\partial f_i / \partial x_j\right]_{\mid x = 0}$ has no eigenvalue with real part equal to zero. Then there exists a neighbourhood $N$ of the equilibrium point and a homeomorphism $h: N \rightarrow \mathbb{R}^n$, such that $h(0)=0$ and such that in the neighbourhood $N$ the system $x_{k+1} = f(x_k, \pi(x_k))$ is topologically conjugate to its linearization $h(x_{k+1}) = A h(x_k)$.
\end{thm}

\noindent We state the result in full generality, however, for the proofs in this section we only a corollary that the system $x_{k+1} = f(x_k, \pi(x_k))$ is locally asymptotically stable if all eigenvalues of the Jacobian at the equilibrium point have negative real part.

\begin{pro}
    Suppose that the optimal controller $\mathbf{u}$ is exponentially stabilizable, then Assumption \ref{Assumption_optimal} holds with $a_V = \frac{M}{1 - e^{-\lambda}}$.
\end{pro}
\begin{proof}
Writing out the expression for the reward and using the definition of exponential stability gives 
    \begin{align*}
        \frac{1}{1 - \gamma} - \sum_{k = 0}^\infty \gamma^k r(\Psi(k, x, \mathbf{u}(k)), \mathbf{u}(k)) &= \sum_{k = 0}^\infty \gamma^k \left(1 - r(\Psi(k, x, \mathbf{u}(k)), \mathbf{u}(k)) \right) \\
        &\leq \sum_{k=0}^\infty 1 - r(\Psi(k, x, \mathbf{u}(k)), \mathbf{u}(k)) \\
        &\leq  \sum_{k=0}^\infty M \|x\|^2 e^{-\lambda k} \\
        &\leq  \frac{M\|x^2\|}{1 - e^{-\lambda}}.
    \end{align*}
    This concludes the proof.
\end{proof}

\noindent The following result will be a first step in establishing $\frac{R}{1 - \gamma} - V_\gamma(x)$ as a Lyapunov function.

\subsubsection{Proof of Proposition 1.}
As a preliminary step we note that there is an error in the statement of Proposition \ref{Proposition_bounds} as it should be $V_\gamma(x) - V_\gamma(y) \leq - c \| x \|^2$. 
\begin{proof}
    Since the rewards can be scaled by an arbitrary constant, we suppose without loss of generality that $R = 1$. We begin by proving (\ref{bound1}) for the optimal value function. For the lower bound, note that since $x^T Q x \geq \lambda_{\min}(Q) \| x \|^2$ and $\frac{1}{1 - \gamma} - V_\gamma(x) \geq r(x, \mathbf{u}(x)) = x^T Q x$ this implies that 
    \begin{align*}
        \frac{1}{1 - \gamma} - V_\gamma(x) \geq \lambda_{\min}(Q) \|x \|^2.
    \end{align*}
    The upper bound follows directly from Assumption \ref{Assumption_optimal} (b). Finally we prove that (\ref{bound2}) holds. Set $y = f(x, \mathbf{u}_1(x))$ as the next state, then by the dynamic programming principle
    \begin{align*}
        V_\gamma(x) = r(x, \mathbf{u}(0)) + \gamma V_\gamma(y)
    \end{align*}
    and thus,
    \begin{align*}
        V_\gamma(x) - V_\gamma(y) &= r(x, \mathbf{u}(0)) - (1 - \gamma)V_\gamma(y) \\
        &= \left( r(x, \mathbf{u}(0)) - 1 \right) + (1 - \gamma)\left(\frac{1}{1 - \gamma} - V_\gamma(y) \right) \\
        &\leq (1 - \gamma) \alpha_2(\|x\|) - (1 - r(x, \mathbf{u}(0))) \\
        &\leq \frac{1 - \gamma}{\gamma} \alpha_2(\|x\|) - x^TQx \\
        &\leq  \left( - \lambda_{\min}(Q) + \frac{1 - \gamma}{\gamma} a_V \right) \|x\|^2 
    \end{align*}
    We have that trajectories are decreasing up to some pertubative term. To deal with this term we set $\gamma^\star = \frac{a_V}{a_V + \lambda_{\min}(Q)}$. This then implies that for any $\gamma \in (\gamma^\star, 1)$ we have that 
    \begin{align*}
        \frac{1 - \gamma}{\gamma} < \frac{\lambda_{\min}(Q)}{a_V + \lambda_{\min}(Q)} \cdot \frac{a_V + \lambda_{\min}(Q)}{a_V} = \frac{\lambda_{\min}(Q)}{a_V}.
    \end{align*}
    Therefore, it follows that there exists $c > 0$ such that $ - \lambda_{\min}(Q) + \frac{1 - \gamma}{\gamma} a_V < -c$ and we conclude that (\ref{bound2}) holds.
\end{proof}
\begin{rem}
    As we saw in the proof, we have that the value function strictly increases along trajectories up to some perturbative term that depends on $\gamma$. In fact, this can be used to show practical stability should Assumption \ref{Assumption_optimal} (b) not be required. Immediately we see that the discounting setting provides challenges for stability as the optimal controller need not be asymptotically stabilizing should this assumption not hold. 
\end{rem}

\subsubsection{Proof of Theorem 1.}
\begin{proof}
    Set $\gamma^\star$ is given in the proof of Proposition \ref{Proposition_bounds} and if necessary adjust the value of $c$ in Proposition \ref{Proposition_bounds} such that $c < a_V$. Denote 
    \begin{align*}
        L_\gamma(x) := \frac{1}{1 - \gamma} - V_\gamma(x) 
    \end{align*}
    which will serve as our Lyapunov function and observe that (\ref{bound2}) can be rewritten as 
    \begin{align*}
         L_\gamma(y) -  L_\gamma(x) \leq -c \| x \|^2.
    \end{align*}
    As $ L_\gamma(x) \leq a_V \|x \|^2$ by Assumption \ref{Assumption_optimal} (b) we see that 
    \begin{align*}
         L_\gamma(y) -  L_\gamma(x) \leq -c \| x \|^2 \leq - \frac{c}{a_V}  L_\gamma(x)
    \end{align*}
    which implies that $ L_\gamma(y) \leq \left(1 - \frac{c}{a_V} \right)  L_\gamma(x)$ where $\left(1 - \frac{c}{a_V} \right) \in (0, 1)$ and so iterative this inequality gives 
    \begin{align*}
        L_\gamma(\Psi(k, x, \pi^\star(k))) \leq \left(1 - \frac{c}{a_V} \right)^k  L_\gamma(x).
    \end{align*}
    This implies that 
    \begin{align*}
        \| \Psi(k, x, \pi^\star(k)) \|^2 \leq \left(1 - \frac{c}{a_V} \right)^k \frac{a_V}{\lambda_{\min}(Q)} \|x\|^2
    \end{align*}
    and setting $\lambda = - \ln( 1 - \frac{c}{a_V})$ gives 
    \begin{align*}
        \| \Psi(k, x, \pi^\star(k)) \|^2 \leq \frac{a_V}{\lambda_{\min}(Q)} \|x\|^2 e^{-\lambda k}.
    \end{align*}
    This proves that the optimal controller is asympotically stabilizing.
\end{proof}
\begin{rem}
    In particular, we have shown that the optimal controller is exponentially stablizing. 
\end{rem}

\subsubsection{Proof of Theorem 2.}

\noindent To break up the length of the proof we first prove the following key lemma.

\begin{lem} \label{value_approximation}
   Suppose that the reward function is defined according to (\ref{reward}) and satisfies the terminal condition (\ref{terminal}) and the neural policy satisfies Assumption \ref{policy}. Let $\pi$ denote the optimal policy to (\ref{objective}). Then for any $\epsilon > 0$ there exists $\delta > 0$ such that $\|\pi - \pi_\phi\|_D < \delta$ implies that $$\|V^{\pi}_\gamma - V^{\pi_\phi}_\gamma\|_D < \epsilon.$$ 
\end{lem}
\begin{proof}
    \textbf{Step 1}. Show that the value function inherits the regularity from the reward function. \newline 
    
    \noindent Note that by definition, 
    \begin{align*}
        V_\gamma^{\pi_\phi}(x) = \sum_{k=0}^\infty \gamma^k r(\Psi(k, x, \pi_{\phi, k}(x), \pi_{\phi, k}(x)) 
    \end{align*}
    is a sum of a sequence of $C^2$ functions $\{f_k(x)\} := \{\gamma^k r(\Psi(k, x, \pi_{\phi, k}(x), \pi_{\phi, k}(x))\}$. Since the reward function is $C^2(D)$ on a compact set and satisfies the terminal condition, it follows from the Weierstrass M-test that $\sum_{k=0}^\infty f_k$ and $\sum_{k=0}^\infty \frac{\partial}{\partial x_j} f_k$ uniformly converge on $D$ and thus from the result on the derivative of uniformly converging functions (reference this later), by taking limits for the derivatives, it follows that 
    \begin{align*}
        \lim_{N \to \infty} \nabla \sum_{k=0}^N f_k = \nabla V_{\gamma}^{\pi_\phi}. 
    \end{align*}
    By making a similar argument for the second order partial derivatives it follows that 
    \begin{align*}
        \lim_{N \to \infty} \nabla^2 \sum_{k=0}^N f_k = \nabla^2 V_{\gamma}^{\pi_\phi}.
    \end{align*}
    It also follows that the value function is $C^2(D)$. \newline 

    \noindent \textbf{Step 2}. Approximate the value function by a neural network policy. \newline 

    \noindent Firstly, by the Bellman dynamic programming principle, we have that 
    \begin{align*}
        V_{\gamma}^\pi(x) &= r(x, \pi(x)) + \gamma V_{\gamma}^\pi (f(x, \pi(x))) \\
        V_{\gamma}^{\pi_\phi}(x) &= r(x, \pi(x)) + \gamma V_{\gamma}^{\pi_\phi} (f(x, \pi(x)))
    \end{align*}
    The idea will be to inductively bound the value function by carefully having the same argument to the value functions to avoid an exponential blow up of the Lipschitz constants. Note that by the assumption on the reward function, there exists $v_r \in D$ such that 
    \begin{align*}
        r(x, \pi_\phi(x)) &= r(x, \pi(x)) + \nabla_u r(x, \pi(x))^T(\pi(x) - \pi_\phi(x)) + \frac{1}{2} (\pi(x) - \pi_\phi(x))^T\nabla^2_u r(v_r, \pi(v_r)) (\pi(x) - \pi_\phi(x)).
    \end{align*}
    The term with the gradient can be bounded by Cauchy-Schwarz, however the quadratic form requires a bit more work. Note that since $\| \pi(x) - \pi_\phi(x) \| \leq \delta$ we have that 
    \begin{align*}
        |\delta^T \nabla^2_u r(v_r) \delta | = |\text{Trace} \left(\nabla^2_u r(v_r, \pi(v_r)) (\pi(x) - \pi_\phi(x)) (\pi(x) - \pi_\phi(x))^T \right)| \leq m \| \nabla^2_u r(v_r, \pi(v_r)) \| \delta^2
    \end{align*}and thus we see that 
    \begin{align*}
        |r(x, \pi_\phi(x)) - r(x, \pi(x))| \leq \| \nabla_u r \| \delta + \frac{m}{2} \| \nabla^2_u r\| \delta^2,
    \end{align*}
    where $\| \nabla_u r \| = \sup_{x \in D} \| \nabla_u r(x, \pi(x)) \|$ and $\| \nabla_u^2 r \| = \sup_{x \in D} \| \nabla_u^2 r(x, \pi(x)) \|$. Similarly, writing the Taylor expansion for the value function, there exists $v_V$ such that  
    \begin{align*}
        V^{\pi_\phi}_\gamma f(x, \pi_\phi(x)) &= V^{\pi_\phi}_\gamma (f(x, \pi(x)) + \nabla V^{\pi_\phi}_\gamma (f(x, \pi(x)))^T \left( f(x, \pi(x)) - f(x, \pi_\phi(x)) \right) \\
        &+ \frac{1}{2} \left( f(x, \pi(x)) - f(x, \pi_\phi(x)) \right)^T \nabla^2 V^{\pi_\phi}_\gamma (f(v_V, \pi(v_V))) \left( f(x, \pi(x)) - f(x, \pi_\phi(x)) \right).
    \end{align*}
    Combining this identity along with the bound for the reward function and the vector field being Lipschitz in the $u$ component we have that
    \begin{align*}
        |V^{\pi}_\gamma f(x, \pi(x)) - V^{\pi_\phi}_\gamma (f(x, \pi(x))| 
        &\leq \gamma | V_{\gamma}^\pi (f(x, \pi(x))) - V^{\pi_\phi}_\gamma (f(x, \pi(x)))| + \| \nabla_u r \| \delta + \frac{m}{2} \| \nabla^2_u r\| \delta^2  \\
        &+ \gamma L_u \| \nabla V^{\pi_\phi} \| \delta + \frac{\gamma(m + n)}{2} L_u^2 \| \nabla^2 V^{\pi_\phi}_\gamma \| \delta^2.
    \end{align*}
    where $\| \nabla V^{\pi_\phi} \| = \sup_{x \in D} \| \nabla V^{\pi_\phi}(f(x, \pi(x))) \|$ and $ \| \nabla^2 V^{\pi_\phi}_\gamma \| = \sup_{x \in D} \| \nabla^2 V^{\pi_\phi}_\gamma (f(x, \pi(x))) \|$. By inductively arguing as we have done above we obtain the following bound
    \begin{align*}
        |V^{\pi}_\gamma f(x, \pi(x)) - V^{\pi_\phi}_\gamma (f(x, \pi(x))| &\leq \lim_{k \to \infty} \left( \gamma^k | V_{\gamma}^\pi (f(x, \pi(x))) - V^{\pi_\phi}_\gamma (f(x, \pi(x)))| \right) +  \left( \delta \| \nabla_u r \| + \frac{m\delta^2}{2} \| \nabla^2_u r\| \right) \sum_{k= 0}^\infty \gamma^k \\
        &+ \left( \delta L_u \| \nabla V^{\pi_\phi} \|^2  + \frac{(m + n)\delta^2}{2}  L_u^2 \| \nabla^2 V^{\pi_\phi}_\gamma \| \right) \sum_{k = 1}^\infty \gamma^k.
    \end{align*}
    The first term in the limit will converge to 0 by Theorem 1. Subsequently, we see that $|V^{\pi}_\gamma (x) - V^{\pi_\phi}_\gamma (x)|$ can be made arbitrarily small so long as we choose $\delta$ sufficiently small. This concludes the proof of the lemma.
\end{proof}

\noindent Now we are ready to prove Theorem \ref{thm_nnet_stab}.

\begin{proof}
    Denote $\max_{\|x \| \leq \delta} \| \pi(x) \| = M_\delta$ which exists since $\pi$ is a smooth function of the state by Assumption \ref{Assumption_policy}. Since $\pi_\phi(0) = 0$ we have that $M_\delta \to 0$ as $\delta \to 0$. Furthermore, there exists $\delta > 0$ such that  
    \begin{align*}
        L_x \delta + L_u M_\delta \leq \Delta. 
    \end{align*}
    Note from the proof of Theorem 1 and (\ref{bound1}) of Proposition \ref{Proposition_bounds} that $L_\gamma(y) < (1 - \frac{c}{2 a_v}) L_\gamma(x)$ holds for all $x \in D \setminus B_\delta(0)$ where $B_\delta(0)$ is a $\delta$ neighborhood of the origin. Then, by the Bolzano-Weierstrass theorem we are guaranteed that the following minimum is strictly positive, 
    \begin{align*}
        \eta := \min_{x \in D \setminus B_\delta(0)} \left(1 - \frac{c}{2 a_v}\right) L_\gamma(x) - L_\gamma(f(x, \pi^\star(x)) > 0.
    \end{align*}
    Therefore, by Lemma \ref{value_approximation}, there exists $\epsilon_1 > 0$ such that if $\| \pi^\star - \pi_\phi\|_D < \epsilon_1$ then $\|V_\gamma - V_\gamma^{\pi_\phi}\|_D < \frac{\eta}{2}$ and thus we see that the following inequality also holds in the place of a neural policy 
    \begin{align*}
        L_{\gamma}^{\pi_\phi}(y) < (1 - \frac{c}{2 a_v}) L_{\gamma}^{\pi_\phi}(x)\; \text{ for all } x \in D \setminus B_\delta(0),
    \end{align*}
    where $L_{\gamma}^{\pi_\phi}(x) := \frac{1}{1 - \gamma} - V_{\gamma}^{\pi_\phi}(x)$. By a similar argument as the one above, there exists $\epsilon_2 > 0$ for which we have that $\| \pi^\star - \pi_\phi\|_D < \epsilon_2$ implies  
    \begin{align*}
        \frac{\lambda_{\min}(Q)}{2} \| x \|^2 \leq L_{\gamma}^{\pi_\phi}(x) \leq 2 a_V \| x \|^2 \; \text{ for all } x \in D \setminus B_\delta(0),
    \end{align*}
    and so we set $\epsilon = \min(\epsilon_1, \epsilon_2)$. Thus, combining the two inequalities implies that 
    \begin{align*}
        \| \Psi(k, x, \pi_\phi(k)) \|^2 \leq \left(1 - \frac{c}{2 a_V} \right)^k \frac{4 a_V}{\lambda_{\min}(Q)} \|x\|^2  \; \text{ for all } x \in D \setminus B_\delta(0). 
    \end{align*}
    This shows that all trajectories under the controller $\pi_\phi$ reach $B_\delta(0)$. However, since the system is in discrete time we could jump arbitrarily far out of this ball. To bound the next state for $x \in B_\delta(0)$ we use the Lipschitz assumption for the vector field. Indeed, by Assumption \ref{Assumption_Lipschitz} we have for any $x \in B_\delta(0)$ that 
    \begin{align*}
        \| f(x, \pi_\phi(x)) \| = \| f(x, \pi_\phi(x)) - f(0, \pi_\phi(0)) \| \leq L_x \delta + L_u M_\delta \Delta. 
    \end{align*}
    This shows that the next state will always be contained in $B_\Delta(0)$ meaning that all trajectories eventually enter $B_\Delta(0)$ and stay in this neighborhood. This proves that the neural policy $\pi_\phi$ is practically stabilizing.
\end{proof}

\subsubsection{Proof of Corollary 1.} 
\begin{proof}
By Theorem \ref{thm_nnet_stab} there exists $\gamma^\star$ and $\epsilon > 0$ such that any neural policy $\pi_\phi$ satisfying $\|\pi_\phi - \pi^\star \| < 2 \epsilon$ is practically stabilizing. Proposition \ref{Proposition_temp} asserts that we can set the temperature parameter $\alpha^\star$ such that for any $\alpha \in (0, \alpha^\star)$ we have that 
\begin{align*}
    |\mathbb{E}\pi^{\text{SAC}} - \pi^\star |_D < \epsilon.
\end{align*}
Therefore, we immediately see that if $\| \mathbb{E} \left[\pi^{\text{SAC}} - \pi_{\phi}{^{\text{SAC}}} \right] \|_D < \epsilon$ then practical stability holds for the policy $\mathbb{E} \left[ \pi_{\phi}{^{\text{SAC}}} \right]$. 
\end{proof}

\subsubsection{Proof of Proposition 2.}
\begin{proof}
    Since the policy is a Gaussian random variable, (\ref{obj_entropy}) is given by 
    \begin{align*}
        \mathbf{J}_\gamma^{\pi_\phi^\text{SAC}} &= \mathbb{E}_{\rho_{\pi_\phi^\text{SAC}}} \sum_{k = 0}^\infty \gamma^k r(\Psi(k, x, \pi_{\phi, k}^\text{SAC}(x), \pi_{\phi, k}^\text{SAC}(x)) + \frac{\alpha}{2} \gamma^k \ln \left( 2 \pi e \sigma(x_k)^2 \right) \\
        &= \sum_{k = 0}^\infty \mathbb{E}_{x_k \sim d_k, a_k \sim \pi(a_k \mid x_k)} \gamma^k r(\Psi(k, x, \pi_{\phi, k}^\text{SAC}(x), \pi_{\phi, k}^\text{SAC}(x)) \\
        &= \sum_{k = 0}^\infty \mathbb{E}_{x_k} \left[ \mathbb{E}_{a_k \mid x_k} \left[ \gamma^k (1 - x^TQx) + \frac{\alpha}{2} \gamma^k  \ln \left( 2 \pi e \sigma(x_k)^2 \right)  \right] \right] \\
        &= \mathbb{E}_\rho \sum_{k=0}^\infty \left[ \gamma^k (1 - x_k^T Q x_k) + \frac{\alpha}{2} \gamma^k  \ln \left( 2 \pi e \sigma(x_k)^2 \right) \right] \\
        &= \frac{1}{1 - \gamma} - \mathbb{E}_\rho \left[ \sum_{k=0}^\infty x_k^T Q x_k \right] + \mathbb{E}_\rho \left[ \sum_{k=0}^\infty \frac{\alpha}{2} \gamma^k  \ln \left( 2 \pi e \sigma(x_k)^2 \right) \right]
    \end{align*}
    where $d_k$ is the marginal distribution of $x_k$ and $\rho$ is the distribution of the trajectory $(x_0, x_1, \hdots)$. 
    Since $\mathbb{E} \pi_\phi^{\text{SAC}} = \pi^\star$ maximizes $\frac{1}{1 - \gamma} - \mathbb{E}_\rho \left[ \sum_{k=0}^\infty x_k^T Q x_k \right]$. Moreover, as $D$ is a compact set, $\sup_{x \in D} \sigma(x)^2 := \sigma_D^2 < \infty$ and we have the following bound
    \begin{align*}
        \mathbb{E}_\rho \left[ \sum_{k=0}^\infty \frac{\alpha}{2} \gamma^k  \ln \left( 2 \pi e \sigma(x_k)^2 \right) \right] \leq \frac{\alpha}{2 (1 - \gamma)} \ln(2 \pi e \sigma_D^2).
    \end{align*}
    Therefore, by Assumption \ref{Assumption_optimal} (a) we have that $\| \mathbb{E} \pi_\phi^{\text{SAC}} - \pi^\star \|_\infty \to 0$ as $\alpha \to 0$. This implies there exists $\alpha^\star > 0$ such that  
    \begin{align*}
        \| \mathbb{E} \pi_\phi^{\text{SAC}} - \pi^\star \|_D < \epsilon
    \end{align*}
    for all $\alpha \in (0, \alpha^\star)$.
\end{proof}
\subsubsection{Proof of Proposition 3.}

Before writing the proof we apologize that there is an error in the main submission for condition (\ref{linear_order}). The correct condition should be 
\begin{equation*}
    | \mathbb{E} \left[\pi^\star(x) - \pi_{\phi}{^{\text{SAC}}}(x) \right] | \leq \eta \|x\|.
\end{equation*}
Moreover, the value of $\eta$ should be $$\eta < \frac{c}{\| \nabla L_\gamma \|^2}$$ where $\| \nabla L_\gamma \|^2 := \sup_{x \in D} \| \nabla L_\gamma(x) \|^2$.
\begin{proof}
    This result will extend the practical stability of Corollary 1 to asymptotic stability. Let $L_\gamma(x)$ be the Lyapunov function of the optimal policy $\pi^\star$ as in Theorem \ref{thm_converge}. The Lyapunov function satisfies $L_\gamma(0) = 0$ and $L_\gamma(x) > 0$ for all $x \in D \setminus \{0\}$. We will show that under the assumptions of the proposition that the same Lyapunov function verifies $L_\gamma(y) - L_\gamma(x) < 0$ for all $x \in D \setminus \{0\}$ where $y$ is the next state of the perturbed system
    \begin{align*}
        y = f(x, \pi^\star(x)) + g(x)
    \end{align*}
    where $g(x) = f(x, \pi_{\phi}^{\text{SAC}}(x)) - f(x, \pi^\star(x))$. By the mean value theorem we have that 
    \begin{align*}
        L_\gamma(f(x, \pi^\star(x)) + g(x)) = L_\gamma(f(x, \pi^\star(x))) + \nabla L_\gamma(v)^T g(x) 
    \end{align*}
    and this implies that 
    \begin{align*}
        L_\gamma(f(x, \pi^\star(x)) + g(x)) - L_\gamma(x) &= L_\gamma(f(x, \pi^\star(x)) + g(x)) - L_\gamma(f(x, \pi^\star(x))) + L_\gamma(f(x, \pi^\star(x))) - L_\gamma(x) \\
        &= \nabla L_\gamma(v)^T g(x) + L_\gamma(f(x, \pi^\star(x))) - L_\gamma(x) \\
        &\leq \| \nabla L_\gamma \|^2 \|g(x)\|^2 - c \|x\|^2 
    \end{align*}
    and thus, we see for the perturbed system that 
    $$
    L_\gamma(f(x, \pi^\star(x)) + g(x)) - L_\gamma(x) < 0.
    $$
    Therefore, by Theorem 1.2 of \cite{bof2018lyapunovtheorydiscretetime} we have that the Lyapunov function $L_\gamma(x)$ verifies asymptotic stability for the system $x_{k+1} =  f(x_k, \pi_{\phi}^{\text{SAC}}(x_k))$ and the neural policy is asymptotically stabilizing.
    
\end{proof}
\begin{rem}
    As was mentioned previously, guaranteeing asymptotic convergence through the policy trained by RL is quite challenging and to handle the limitations of practical stability, an algorithm likes ours should be used.
\end{rem}

\subsubsection{Proof of Theorem 3.} The idea of the proof entails using the neural policy to guide the system into a sufficiently small neighborhood of the equilibrium point. Once the system is within this region, linearization is applied to ensure local asymptotic stability. The following result is needed for the proof.

\begin{proof}
   The linearization for the vector field $f(x, \pi_\phi^{\text{SAC}}(x))$ at the origin is given by
    \begin{align*}
        A + B \hat{K} := \frac{\partial f}{\partial x}((0, \pi_\phi^{\text{SAC}}(0)) + \frac{\partial f}{\partial u}((0, \pi_\phi^{\text{SAC}}(0)) \frac{\partial \pi}{\partial x}(0) = \frac{\partial f}{\partial x}((0, \pi_\phi^{\text{SAC}}(0)) + \frac{\partial f}{\partial u}((0, \pi_\phi^{\text{SAC}}(0)) \hat{K},
    \end{align*} 
    \noindent where $\hat{K}$ is the Jacobian of the neural policy. Let $K$ be the optimal gain matrix of the actual system as a result of solving the DARE. By the Hartman-Grobman theorem there exists $\delta > 0$ such that the system $f(x, \pi_\phi^{\text{SAC}}(x))$ is topologically conjugate to its linearization. By Assumption \ref{Assumption_LQR}, we have that our learning algorithm converges to an estimate $\hat{K}$ such that 
    $$\|K - \hat{K}\|_2 < \frac{\epsilon_K}{\|B\|_2}.$$
    \noindent Moreover, by Assumption \ref{Assumption_actorloss} and Corollary 1, our learning algorithm converges to a neural policy such that is practically stabilizing to the neighborhood $B_\delta(0)$. Note that since the eigenvalues of a matrix are continuous functions of the entries of the matrix, there exists $\epsilon_K > 0$ such that for all matrices $M$ satisfying $\|A + BK - M\|_2 < \epsilon_K$, the eigenvalues $\lambda_1(M), \hdots, \lambda_m(M)$ all have negative real part. Since
    \begin{align*}
        \| A + B K - (A + B \hat{K}) ||_2 \leq \|B\|_2 \|K - \hat{K} \| \leq \epsilon_K.
    \end{align*}
    we have that the linearization for the system with the neural controller is Hurwitz. This proves local asymptotic stability of the neural policy and combined with the fact that $B_\delta(0)$ is reachable for all $x \in D$, this implies that the neural policy $\pi_\phi^{\text{SAC}}$ is asymptotically stabilizing.   

\end{proof}

\subsubsection{Alternative Result to Theorem 2.}
To end the section, as stated in Remark 3 we prove an alternative result involving different assumptions that is interesting in its own light. We provide an informal discussion in the hopes that it can lead to future work. Let $\Pi_{\text{NNet}}$ denote the space of neural networks. 

\begin{pro} \label{agent_reach}
    Fix $\epsilon > 0$ and a consider a compact state space $D$. Let $\pi^\star \in \Pi_{\text{NNet}}$ be the optimal policy that maximizes the deterministic RL algorithm over $\Pi_{\text{NNet}}$ with reward 
    $$r(x_n, u_n) = R - x_n^T Q x_n$$
    where $Q$ is a positive definite matrix and suppose that the system $x_{k+1} = f(x_k, u_k)$ is stabilizable with polynomial convergence rate 
    \begin{equation}
        \|x_n \| \leq \frac{C}{(n + 1)^p}
    \end{equation}
    for some $p > 0$ and $C > 0$. Furthermore, suppose that the discount factor $\gamma$ satisfies 
    \begin{equation}
        \left( \frac{1 - \lambda_{\min(Q)} \epsilon^2}{1 - \lambda_{\min(Q)} \epsilon^2/4} \right)^\kappa < \gamma \leq 1
    \end{equation}
        where $$\kappa = \left( \frac{\bar{\lambda} \epsilon}{2C} \right)^p$$ and $\bar{\lambda} = \sqrt{\frac{\lambda_{\min(Q)}}{\lambda_{\max(Q)}}}$. Then there exists $N \in \mathbb{N}$ such that $\| x_N \| < \epsilon$.
\end{pro}

\begin{proof}
    As a first step, we show that if the system is stabilizable then there exists a neural policy $\pi_\phi$ such that $f(x, \pi_\phi(x))$ is locally asymptotically stable. Denote the stabilizing controller as $\kappa$, the linearization at the origin is given by 
    $$
    \nabla f^{\pi^\star}_{\mid x = 0} = A + B \nabla \kappa
    $$
    where $A + B \nabla \kappa$ is Hurwitz. By the smooth version of the universal approximation theorem on the compact set $D$ (which can be extended to vector value functions by placing the components of the function in a parallel architecture for the neural network) there exists a sequence of neural networks $\pi_{\phi}^n$ such that $\| \nabla \kappa - \nabla \pi_{\phi}^n \|_D \to 0$. This implies that 
    \begin{align*}
        \| \nabla f^{\kappa}_{\mid x = 0} - \nabla f^{\pi_n}_{\mid x = 0}  \| &= \| B (\nabla \kappa - \nabla \pi_n) \| \\
        &\leq \| B \| \| \nabla \kappa - \nabla \pi_n \| \to 0
    \end{align*}
    as $n \to \infty$. Thus, it follows that the spectra converge
    $$
    \sigma(\nabla f^{\pi_n}_{\mid x = 0}) \to \sigma(\nabla f^{\kappa}_{\mid x = 0})
    $$
    and there exists a neural policy $\pi_\phi$ such that $A + B \nabla \pi_\phi$ is Hurwitz and the system $f(x, \pi_\phi(x))$ is asymptotically stable in some neighborhood of the origin. From now on, we restrict $x$ to this neighborhood as we can always use the previous results to formulate global guarantees. \newline 

    \noindent Without loss of generality we can set $R = 1$ and rescale $Q$ such that $1 - x^T Q x \geq 0$ for all $x \in D$.Now, suppose for a contradiction that there exists $\varepsilon > 0$ such that $\Psi(k, x, \pi^\star_k(x)) \notin B_\varepsilon(0)$ for all $k$. This implies that 
    $$
    J^{\pi^\star}(x_0) \leq \frac{1 - \lambda_{\min}(Q) \epsilon^2 }{1 - \gamma}.
    $$
    
    \noindent Let $\{ x_k \}_{n=1}^\infty$ be the trajectory generated under the neural policy $\pi_\phi$ and since it is stabilizing we also suppose that it satisfies the same polynomial convergence it approximates the derivatives of the stabilizing policy. By the polynomial convergence assumption we denote $N$ as the first time step such that the following holds 
    $$\| x_N \| \leq \frac{\lambda_{\min}}{\lambda_{\max}} \epsilon.$$ 
    
    \noindent In particular we have that $N \geq \left( \frac{2C}{\bar{\lambda} \epsilon}\right)^\frac{1}{p}$.
    Under this policy we can find a lower bound for the objective function 
    \begin{align*}
        J^{\pi_\phi}(x_0) &= \sum_{k=1}^{N-1} \gamma^n (1 - x_k^T Q x_k) + \sum_{k=N}^{\infty} \gamma^n (1 - x_k^T Q x_k) \\
        &> \sum_{n=N}^{\infty} \gamma^n (1 - x\lambda_{\max}(Q) \frac{\epsilon^2}{4}) \\
        &= \gamma^N \frac{(1 - \lambda_{\max}(Q)\frac{\epsilon^2}{4}}{1 - \gamma}
    \end{align*}
    and so substituting the bounds for $N$ and $\gamma$ will show that 
    $$
    J^{\pi_\phi} > J^{\pi^\star}
    $$
    which contradicts the optimality assumption for $\pi^\star$. 
\end{proof}

\begin{rem}
This result can be combined with Assumption \ref{Assumption_LQR} to prove asymptotic convergence. The idea is to reach the region of attraction of the learned gain matrix that can be learned using our algorithm.
\end{rem}

\newpage
\subsection{Stabilization of our learning algorithm}
The figures below show that the algorithms presented in the numerical experiments section are also capable of learning an asymptotically stabilizing policy. Here we plot fifty trajectories of the positional components where the initial states are the same as the ones we use for the cost comparisons in Figure \ref{fig_costs}. We can also infer that the velocities must approach zero as the positional components remain at zero.

\begin{figure*}[h]
    \centering
    \begin{subfigure}[b]{0.45\textwidth}
        \centering
        \includegraphics[width=\textwidth]{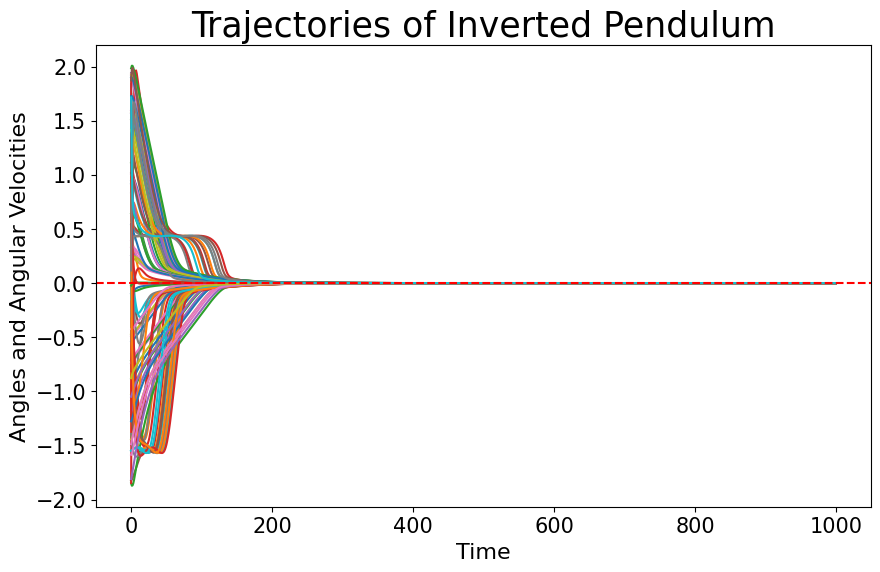}
        \caption{Inverted Pendulum}
        \label{fig:2a}
    \end{subfigure}
    \hfill
    \begin{subfigure}[b]{0.45\textwidth}
        \centering
        \includegraphics[width=\textwidth]{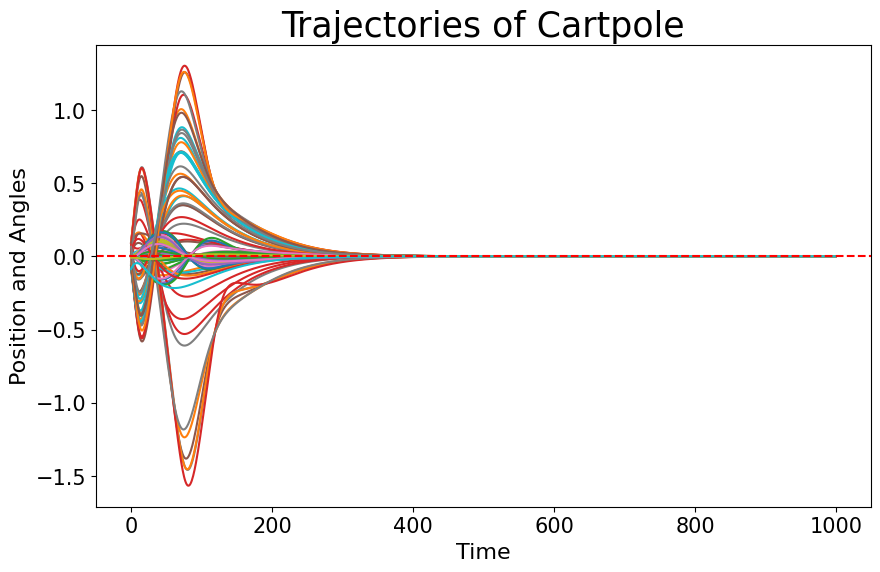}
        \caption{Cartpole}
        \label{fig:2b}
    \end{subfigure}
    
    \vspace{0.5cm}

    \begin{subfigure}[b]{0.45\textwidth}
        \centering
        \includegraphics[width=\textwidth]{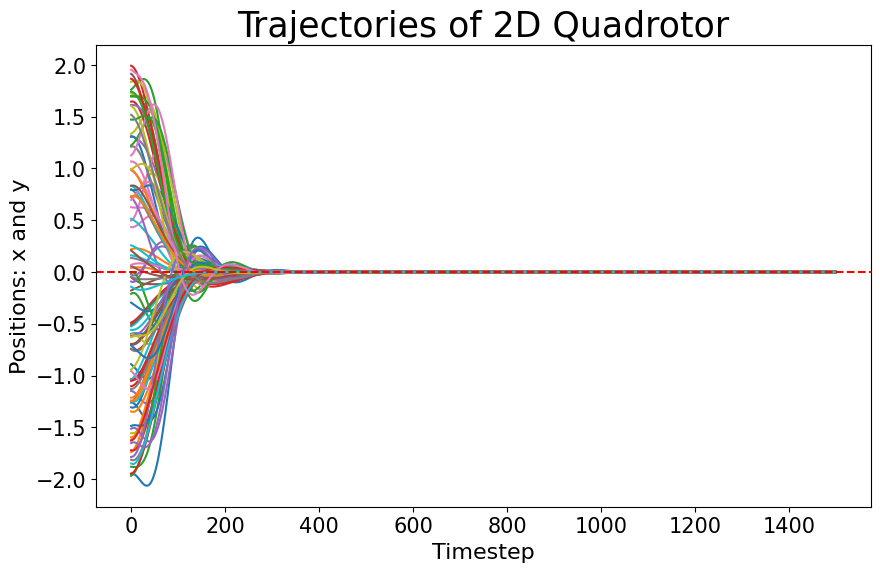}
        \caption{2D Quadrotor}
        \label{fig:2c}
    \end{subfigure}
    \hfill
    \begin{subfigure}[b]{0.45\textwidth}
        \centering
        \includegraphics[width=\textwidth]{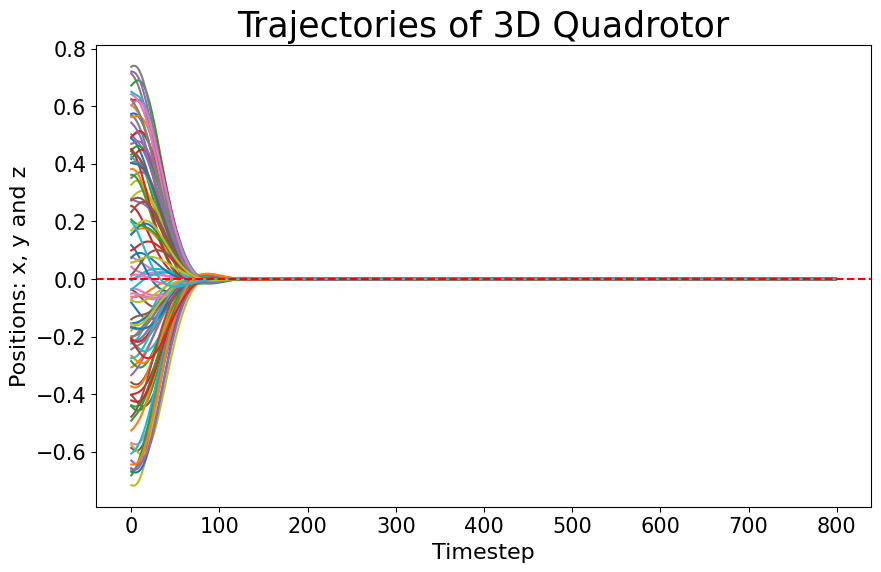}
        \caption{Placeholder for 3D quadrotor}
        \label{fig:2d}
    \end{subfigure}
    
    \caption{Stabilizing Trajectories for Our Algorithm over the Four Environments}
    \label{fig_stabilization}
\end{figure*}

\newpage
\subsection{Complete Learning Algorithm}

\begin{algorithm}
\caption{Soft Actor Critic with LQR for Unknown Dynamics}
\label{algorithm}
\begin{algorithmic}[1]
\State Initialize neural network parameters $\psi, \bar{\psi}, \theta, \phi$ and replay buffer $\mathcal{D}$
    \State Set input dimension ($n+m$), output dimension ($n$), discount factor ($\gamma$), entropy coefficient ($\alpha$), distance tolerance ($\eta$), positive definite matrices ($Q$ and $R$), reward scale, learning rate of each neural network
    \State Set a condition for updating $A$ and $B$
    \State $A, B, K \leftarrow $ Zero matrix, $K_{\text{available}} \leftarrow \textbf{False} $ 
    \For{each iteration}
    \For{each environment step}.
        \If{$\| x_{k} \| < \eta $}
        \State Take a zero action or non-zero action $\mathbf{a}_k \sim \pi_\phi\left(\mathbf{a}_k \mid \mathbf{x}_k\right)$
        \State Store $(\mathbf{x}_k, \mathbf{a}_k)$ for (\ref{learnA}, \ref{learnB})  
        \Else 
        \State $\mathbf{a}_k \sim \pi_\phi\left(\mathbf{a}_k \mid \mathbf{x}_k\right)$
        \EndIf
        \State $\mathbf{x}_{k+1} \sim p\left(\mathbf{x}_{k+1} \mid \mathbf{x}_k, \mathbf{a}_k\right)$
        
        \State $\mathcal{D} \leftarrow \mathcal{D} \cup\left\{\left(\mathbf{x}_k, \mathbf{a}_k, r\left(\mathbf{x}_k, \mathbf{a}_k\right), \mathbf{x}_{k+1}\right)\right\}$ 
    \EndFor
    \If{Update $A$ satisfied}
    \State Update $A$ according to (\ref{learnA})
    \EndIf
    \If{Update $B$ Satisfied}
    \State Update $B$ according to (\ref{learnB})
    \State Calculate $K$ according to $A, B$ and the DARE \Comment{$A$ is always updated before $B$}
    \State $K_{\text{available}} \leftarrow \textbf{True} $ 
    \EndIf
    
    \For{each gradient step}
    \State $\psi \leftarrow \psi-\lambda_V \hat{\nabla}_\psi J_V(\psi)$
    \State $\theta_i \leftarrow \theta_i-\lambda_Q \hat{\nabla}_{\theta_i} J_Q\left(\theta_i\right) \text { for } i \in\{1,2\}$
    \If{$K_{\text{available}}$}
    \State $\phi \leftarrow \phi-\lambda_\pi \hat{\nabla}_\phi J(\phi)$ where the actor objective is given by (\ref{gain_actor_loss})
    \Else
    \State $\phi \leftarrow \phi-\lambda_\pi \hat{\nabla}_\phi J_\pi(\phi)$
    \EndIf
    \State $\bar{\psi} \leftarrow \tau \psi+(1-\tau) \bar{\psi}$
    \EndFor

\EndFor

\end{algorithmic}
\end{algorithm}

\newpage

\subsection{Additional Details on SAC and Verification of Stability}

\subsubsection{Data Collection for LQR} Since collecting data while the agent navigates the environment can be a time consuming process, in many cases it is more efficient to collect state measures to get a baseline estimate of the gain matrix and improve this estimate as the algorithm runs. 

\subsubsection{Cost Matrices} 
The positive semidefinite matrices $Q$ and $R$ are always diagonal matrices. Please refer to the test files in the code for the supplementary material to see the choice of $Q$ and $R$. 

\subsubsection{Target network update} 
Since SAC deals with continuous state spaces, it is common in this setting to use a separate target value function that tracks the actual value function to improve the stability of the learning algorithm. As in \cite{Lillicrap2015-pt}, to update the target network, we use an exponentially moving average with smoothing constant $\tau$. The possible values of $\tau$ range from zero to one, where a value of zero means no update to the target value network and a range of one means a complete copy of the current weights. We found that the same value of $\tau$ as in the original SAC paper was suitable for all the stabilization tasks \cite{Haarnoja2018-ha}.  

\subsubsection{Reward Scale}
SAC is dependent on the scaling of the reward signal since this parameter affects the temperature parameter that controls the stochasticity of the policy. We found that a reward scale of two generally works well for our stabilization tasks. 

\subsubsection{Hyperparameters}
Table \ref{table_hyper} below lists the hyperparameters that are common to all environments that are used in the comparative evaluation in Figure \ref{fig_costs}. We choose the hyperbolic tangent as the activation function as we believe that for continuous control tasks, the policy should vary smoothly with respect to the state and it also satisfies the assumptions of Theorem \ref{thm_nnet_stab}. Note that since PPO is an on-policy algorithm, we do not implement a replay buffer for that algorithm. Additionally, Table \ref{table_networks} lists the size of the neural networks for all stabilization tasks. We note that since stabilization requires precise actions near the origin and the policy may have to learn a potentially large gain matrix at the origin, a small network may not generalize to large domains as we have in Figure \ref{fig_stabilization}. Therefore, we employ larger networks for greater approximation ability. 

\begin{table}[h]
    \centering
    \caption{SAC Hyperparameters}
    \label{table_hyper}
    \begin{tabular}{ll}
        \toprule
        Parameter & Value \\
        \midrule
        \textbf{Shared} & \\
        optimizer & Adam \cite{DBLP:journals/corr/KingmaB14} \\
        learning rate & $3 \cdot 10^{-4}$ \\
        discount ($\gamma$) & 0.99 \\
        replay buffer size & $10^6$ \\
        number of samples per minibatch & 256 \\
        activation & Hyperbolic Tangent \\
        \midrule
        \textbf{SAC} & \\
        target smoothing coefficient ($\tau$) & 0.005 \\
        target update interval & 1 \\
        gradient steps & 1 \\
        reward scale & 2 \\
        \bottomrule
    \end{tabular}
\end{table}

\begin{table}[h]
    \centering
    \caption{SAC Environment Specific Networks}
    \label{table_networks}
    \begin{tabular}{llll}
        \toprule
        Environment & Action Dimensions & Number of Hidden Layers & Hidden Layer Neurons \\
        \midrule
        Inverted Pendulum & 1 & 2 & 16 \\
        Cartpole & 1 & 2 & 256 \\
        Quadrotor-2D & 2 & 4 & 256 \\
        Quadrotor-3D & 4 & 5 & 256 \\
        \bottomrule
    \end{tabular}
\end{table}

\subsection{Verification for stability}
We verify the stability of the inverted pendulum with the quadratic Lyapunov function $V =  x^T P x = 1.02512 x_1^2 + 0.0216179 x_2^2 + 0.0379687 x_1  x_2$, obtained by linearizing the continuous-time system. The matrix $P$ is computed by solving the Lyapunov equation $A^T P + PA = -Q$ by setting $Q = I$. As illustrated in Fig.~\ref{fig:roa}, the red ellipse is the ROA estimate, while the green and blue curves are the level sets of the Lyapunov function.

\begin{figure}
    \centering
    \includegraphics[width=0.85\linewidth]{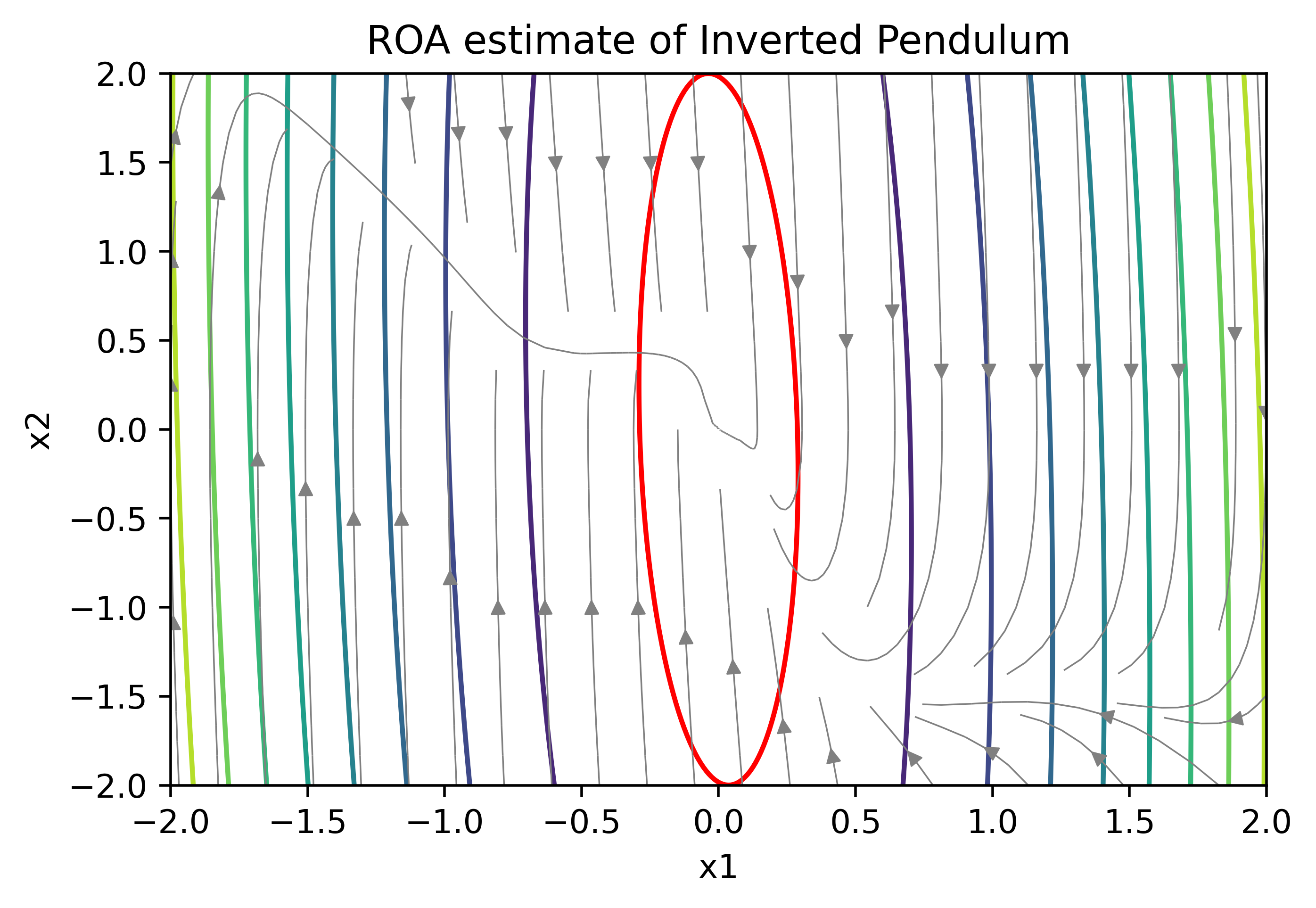}
    \caption{Verified ROA estimate for the inverted pendulum case, where the red curve is the ROA estimate.}
    \label{fig:roa}
\end{figure}

\subsection{Limitations and Future Work} 
\textbf{Convergence Analysis:} In the convergence analysis we assumed that the RL algorithm is capable of learning a near optimal policy. An interesting direction for future research could be to establish conditions for convergence in the neural network setting and to derive results for convergence rates or finite sample approximation guarantees. \newline 

\noindent \textbf{Practical Applications:} An interesting research direction would be to assess the performance improvements of our stabilizing controllers for practical safety critical tasks. \newline 

\noindent \textbf{Verification:} With LyZNet, we can only verify the asymptotic stability for the low-dimensional systems with the learned neural policy, due to the poor scalability of the SMT solver. In the future, we will study how to verify the stability for higher-dimensional systems, given the complexity of the neural network controller. 

\end{document}